\documentclass[11pt]{amsart}
\usepackage[ascii]{inputenc}
\usepackage{fullpage,amssymb}
\usepackage{mathrsfs}
\usepackage{xcolor}
\usepackage{algorithmicx}
\usepackage{algpseudocode}
\usepackage{tikz-cd}
\usepackage{enumitem}
\usepackage{xcolor}
\usepackage[bookmarksnumbered,linktocpage,hypertexnames=false,colorlinks=true,linkcolor=blue,urlcolor=blue,citecolor=blue,anchorcolor=green,pdfusetitle]{hyperref}
\usepackage{mathtools}
\setlist[description]{labelindent=\parindent,itemsep=.2em}

\DeclareMathOperator{\supp}{supp}
\DeclareMathOperator{\esupp}{e-supp}

\DeclareMathOperator{\stab}{stab}

\DeclareMathOperator{\GL}{GL}
\DeclareMathOperator{\SL}{SL}

\DeclareMathOperator{\ST}{ST}

\DeclareMathOperator{\poly}{poly}
\DeclareMathOperator{\Ss}{S}
\DeclareMathOperator{\Mat}{Mat}
\DeclareMathOperator{\Lie}{Lie}

\makeatletter
\newcases{lrdcases}
  {\quad}
  {\hfil$\m@th\displaystyle{##}$\hfil}
  {\hfil$\m@th\displaystyle{##}$\hfil}
  {\lbrace}
  {\rbrace}
\makeatother

\newtheorem{theorem}{Theorem}[section]
\newtheorem{corollary}[theorem]{Corollary}
\newtheorem{lemma}[theorem]{Lemma}
\newtheorem{proposition}[theorem]{Proposition}

\theoremstyle{definition}
\newtheorem{definition}[theorem]{Definition}
\newtheorem{remark}[theorem]{Remark}

\newtheorem{problem}[theorem]{Problem}

\newtheorem{algorithm}[theorem]{Algorithm}
\numberwithin{equation}{section}

\newcommand{\QQ}{\mathbb{Q}}
\newcommand{\CC}{\mathbb{C}}
\newcommand{\ZZ}{\mathbb{Z}}
\newcommand{\RR}{\mathbb{R}}
\newcommand{\NN}{\mathbb{Z}_{\geq 0}}

\begin{document}
\title{Polynomial time algorithms in invariant theory \texorpdfstring{\\}{} for torus actions}
\author{Peter B\"{u}rgisser}
\address{Institut f\"ur Mathematik, Technische Universit\"at Berlin}
\email{pbuerg@math.tu-berlin.de}
\author{M. Levent Do\u{g}an}
\address{Institut f\"ur Mathematik, Technische Universit\"at Berlin}
\email{dogan.mlevent@gmail.com}
\author{Visu Makam}
\address{Institute for Advanced Study, Princeton and School of Mathematics and Statistics, University of Melbourne}
\email{visu@umich.edu}
\author{Michael Walter}
\address{Korteweg-de Vries Institute for Mathematics, Institute for Theoretical Physics, and Institute for Logic, Language and Computation, University of Amsterdam}
\email{m.walter@uva.nl}
\author{Avi Wigderson}
\address{Institute for Advanced Study, Princeton}
\email{avi@ias.edu}
\begin{abstract}
An action of a group on a vector space partitions the latter into a set of orbits.
We consider three natural and useful algorithmic ``isomorphism'' or ``classification'' problems, namely, \emph{orbit equality}, \emph{orbit closure intersection}, and \emph{orbit closure containment}.
These capture and relate to a variety of problems within mathematics, physics and computer science, optimization and statistics.
These orbit problems extend the more basic null cone problem, whose algorithmic complexity has seen significant progress in recent years.

In this paper, we initiate a study of these problems by focusing on the actions of commutative groups (namely, tori).
We explain how this setting is motivated from questions in algebraic complexity, and is still rich enough to capture interesting combinatorial algorithmic problems.
While the structural theory of commutative actions is well understood, no general efficient algorithms were known for the aforementioned problems.
Our main results are polynomial time algorithms for all three problems.
We also show how to efficiently find separating invariants for orbits, and how to compute systems of generating rational invariants for these actions (in contrast, for polynomial invariants the latter is known to be hard).
Our techniques are based on a combination of fundamental results in invariant theory, linear programming, and algorithmic lattice theory.
\end{abstract}
\maketitle
\tableofcontents

\section{Introduction} \label{sec:intro}
Consider the following two problems, which on the face of it have nothing to do with each other:
\begin{enumerate}
\item Will the cue ball's trajectory on a billiards table ever end up in a pocket?
\item Given a bipartite graph $G$, and two functions $w$, $w'$ assigning weights to edges, is it the case that they assign the same weight to \emph{every} perfect matching~$M$ of~$G$?
\end{enumerate}
Both turn out to be orbit problems for torus actions, and exemplify the class of problems we study in this paper.

As our introduction is somewhat long, we break it up as follows.
We start with general background to algorithmic invariant theory in~\S\ref{subsec:alg inv th} and discuss general orbit problems in~\S\ref{subsec:orbit}.
In~\S\ref{subsec:intro-main-results} we define torus actions, discuss our main results, and explain their motivation
from the perspective of algebraic complexity.
In~\S\ref{subsec:applications}, we give examples of how these orbit problems for torus actions arise in and capture natural problems in physics and optimization.
In~\S\ref{subsec:organization}, we discuss the organization of the paper and logical structure of our results.

\subsection{Algorithms in invariant theory}\label{subsec:alg inv th}

Computational invariant theory is a subject whose origins can be traced back to
``masters of computation'' in the 19th century such as Boole, Gordan, Sylvester and Cayley among others.
The second half of the 20th century injected a major impetus to both structural and  computational aspects of these mathematical areas.
On the one hand, the advent of digital computers allowed mathematicians means to study much larger such algebraic structures than could be accessed by hand.
On the other, the parallel development of computational complexity provided a mathematical theory with precise computational models for algorithms and
their efficiency analysis.
This combination has injected many new ideas and questions into invariant theory and related fields, leading to the development of algorithmic techniques
such as Gr\"obner bases and many others, which supported faster and faster algorithms.
Texts on this large body of work can be found, for example, in the books~\cite{DK,Sturmfels,Cox-Little-Oshea}.
While the computational complexity put focus on polynomial time as the staple of efficiency, it also provided means to argue
the likely impossibility of such fast algorithms for certain tasks, through the Cook-Karp-Levin theory~\cite{cook:71,karp:72,levin:73}
of NP-completeness (for Boolean computation) and Valiant's theory of VNP-completeness.

More recently, a further surge in collaboration between algebraists and complexity theorists on these algorithmic questions in invariant theory
and representation theory arose from two (related) sources starting in the turn of this century.
Both imply that these very algorithmic questions in algebra are deeply entwined with the core complexity questions of P vs.~NP  and VP vs.~VNP.
Not surprisingly, new enriching connections between these two research directions are newly found as they develop, providing an exciting collaboration.

The first source is Mulmuley and Sohoni's Geometric Complexity Theory (GCT)~\cite{mulmuley2001geometric}, which highlights
the inherent symmetries of complete problems of these complexity classes, and through these suggests concrete invariant theoretic and
representation theoretic attacks on the questions above.
This has lead to many new questions, techniques, and much faster algorithms (see, for example, \cite{GCTV,FS,BCMW:17,MW19}).

The second source is the work of Impagliazzo and Kabanets~\cite{KI04}, using Valiant's completeness theory for VP and VNP to again attack these major
complexity problems directly through the development of efficient deterministic algorithms for the basic PIT (Polynomial Identity Testing) problem.
This problem, which (again, thanks to Valiant's completeness) has natural symmetries, is very similar to basic invariant theory problems.
Major progress was recently made on resolving such related algorithmic problems, starting with~\cite{Gur04,GGOW16,IQS,IQS2,DM-poly}.
Many others continue to follow, see, for example, \cite{DM-oc,AZGLOW,GGOW:20,BGOWW,BFGOWW,BLNW:20,BFGOWW2}.
We refer to \cite{BFGOWW2} for a recent description of the state-of-art.

\subsection{Orbit problems}\label{subsec:orbit}
We now briefly describe the basic setting and problems of interest, postponing some of the technical details to later sections for the sake of brevity.
A group homomorphism $\rho\colon G \rightarrow \GL(V)$, where $V$ is a vector space (always complex and finite-dimensional) is called a representation of~$G$.
One can think of this as a (linear) action of~$G$ on~$V$, i.e., a map $G \times V \rightarrow V$ where $(g,v) \mapsto \rho(g) v$ satisfies the usual axioms of a group action.
For us, groups will always be algebraic and representations rational, that is, morphisms of algebraic groups.
We will denote $\rho(g) v$ by~$gv$ or~$g \cdot v$.

For $v \in V$, we define its {\em orbit} $O_v := \{gv \ |\ g \in G\}$ (denoted $O_{G,v}$ if the group is not clear from context) to be the
subset of points that can be reached from $v$ by applying a group element.
We denote by $\overline{O_v}$ the topological closure of~$O_v$.
These notions are extremely basic and in many concrete instances very familiar.
One simple example is the action of $\GL_n \times \GL_n$ on $n \times n$ matrices by left and right multiplication:
clearly, the orbit of a matrix~$A$ consists of the matrices having the same rank as~$A$;
moreover, the orbit closure of~$A$ is the set of matrices whose rank is at most the rank of $A$.
Another example is the conjugation action of $\GL_n$ on $n \times n$ matrices, where the orbits are characterized by Jordan normal forms.%
\footnote{The orbit closures of two matrices intersect if and only if the matrices have the same eigenvalues (counted with multiplicity).}

Understanding the space of orbits of a given group action is perhaps the most basic task of invariant theory. The following three basic algorithmic problems will be the focus of this paper.

\begin{problem}\label{prob:main}
Let $\rho \colon G \rightarrow \GL(V)$ be a representation of a group $G$. Given $v,w \in V$:
\begin{enumerate}
\item \textbf{Orbit equality:} Decide if $O_v = O_w$;
\item\label{prob:oci} \textbf{Orbit closure intersection:} Decide if $\overline{O_v} \cap \overline{O_w} \ne \emptyset$;
\item \textbf{Orbit closure containment:} Decide if $w \in \overline{O_v}$.\footnote{The special case of $w = 0$ is called the null cone membership problem. In fact, many of the recent algorithmic advances mentioned above efficiently solve the null cone problem for specific group actions, see \cite{BFGOWW2} and references therein. The motivation of this paper is to extend that understanding to these more general problems.}
\end{enumerate}
\end{problem}

As we will discuss the computational complexity of algorithms for these problems, one needs to specify how inputs are given and how we measure their size.
We will discuss this, but for now it suffices to think of $n=\dim(V)$, the degree of~$\rho$ (assuming it is a polynomial function), and the bit-length of the input vectors $v,w$ as the key size parameters.

The aforementioned problems capture and are related to a natural class of ``isomorphism'' or ``classification'' problems across many domains in mathematics, physics and computer science.
Examples include the graph isomorphism problem~\cite{Derksen-graph}, non-commutative rational identity testing~\cite{GGOW16,IQS2}, equivalence problems on quiver representations~\cite{DM-arbchar,DM-si}, matrix and tensor scaling~\cite{BGOWW,BFGOWW}, classification of quantum states~\cite{bennett1993teleporting} and module isomorphism problems~\cite{BL08}.

To briefly hint at the role of invariant theory, let us take a closer look at problem~(\ref{prob:oci}), that is, the problem of orbit closure intersection.
We denote by $\CC[V]$ the ring of polynomial functions on $V$.
A polynomial function~$f$ on~$V$ is called \emph{invariant} if it is constant along orbits, i.e., $f(gv) = f(v)$ for all~$g \in G$ and~$v \in V$.
The collection of all invariant polynomials forms a subring $\CC[V]^G$, called the \emph{invariant ring}.
Since polynomials are continuous, invariant polynomials are constant along orbit closures.
In particular, two points $v$ and $w$ are indistinguishable by invariant polynomials when their orbit closures intersect.
Amazingly, the converse is also true for a large class of group actions thanks to a result due to Mumford: if the orbit closures of~$v$ and~$w$ do not intersect, then they can always be distinguished by an invariant polynomial.
See Theorem~\ref{thm:Mumford} for a precise statement.

Mumford's theorem suggests an approach to orbit closure intersection -- test if $f(v) = f(w)$ for all invariant polynomials $f$.
For this strategy to be effective, one needs a computational handle on invariant polynomials.
Naively there are infinitely many polynomials, but a foundational result of Hilbert helps tackle this issue.
A \emph{system of generating polynomial invariants} is a collection of invariant polynomials~$\{f_1,\dots,f_r\}$ such that any other invariant polynomial can be written as a polynomial in the~$f_i$'s.
In particular, to test for orbit closure intersection it suffices to test whether each of the~$f_i$ take the same value on both points.
Hilbert showed the existence of a finite system of generating polynomial invariants and also gave an algorithm to produce them~\cite{Hilb1}.
Since then, many improvements on the complexity of such algorithms were developed, but even today this task is, in general, infeasible.
One basic obstacle is the very description of such a system of generating invariants, coming both from the size of this set and the degree of each polynomial in it.

Nearly a century later, a (singly) exponential bound (in~$n$) on the degrees of a system of generating polynomial invariants was achieved for a very general class of group actions~\cite{Der01}, which is unfortunately the best possible in this generality, see~\cite{DM-exp}.
A singly exponential bound is necessary to capture a polynomial with a poly-sized (in $n$) arithmetic circuit, but is by no means sufficient.%
\footnote{For example, the permanent of an $n\times n$ matrix, which has degree~$n$, is believed to require exponential circuit size.
This is essentially the content of Valiant's proof that the permanent is complete for the class VNP, combined with the hypothesis that VNP~$\ne$~VP.}
Another issue that one has to deal with is the number of invariants in a system of generating polynomial invariants, and it is often the case that there are exponentially many in any system.%
\footnote{This is already the case for the matrix scaling action discussed in Section~\ref{subsec:applications}.}
This led Mulmuley \cite{GCTV} to suggest the notion of a \emph{succinct circuit} as a way to capture a system of generating polynomial invariants with a view towards using them for orbit closure intersection.
Unfortunately, this approach does not seem to be computationally feasible either.
See~\cite{GIMOWW} where Mulmuley's conjecture~\cite[Conjecture~5.3]{GCTV} on the existence of succinct circuits was disproved under natural complexity assumptions.
What is perhaps most surprising is that this already happens for a \emph{commutative} group action, namely when~$G$ is a torus.
Further, an example of a group action was given where any system of generating polynomial invariants must contain a VNP-hard polynomial.

The negative result above seem to suggest that the algorithmic tasks at hand are infeasible, even for torus actions, i.e., groups of the form $(\CC^\times)^d$.
The main results of our paper show the opposite: \emph{all of them are efficiently solvable for torus actions!}

The main novelty on our approach is using {\em rational invariants} instead of polynomial invariants.
A rational invariant is a quotient of polynomials that is invariant, see Section~\ref{subsec:intro-main-results} for a precise definition.
This is a bit unexpected since Mumford's theorem simply does not extend to rational invariants: it is easy to construct examples where two points whose orbit closures intersect are distinguished by a rational invariant.
Yet, for representations of tori, we show that (a certain special collection of) rational invariants can be used (in a delicate way) to capture not just orbit closure intersection, but orbit closure containment and orbit equality as well.
Moreover, we show that rational invariants are computationally easy in this case, in stark contrast with the aforementioned hardness results for polynomial invariants~\cite{GIMOWW}.

Inspired by the connections to the P vs.~NP problem, the GCT program makes several predictions in invariant theory.
The setting in which most of the predictions and conjectures are formulated is the setting of rational representations of connected reductive groups
(which we will define later). Here, we want to point out that among connected reductive groups, the class of commutative groups happen to be precisely tori.
Thus, our main results should be viewed as conclusively verifying several predictions of GCT in the commutative case.
Moreover, the barrier result on the computational efficiency of polynomial invariants \cite{GIMOWW} along with our results on rational invariants suggest
that a more thorough investigation of rational invariants is needed in the case where the acting group is non-commutative, e.g., $\SL_n$.



\subsection{Torus actions and main results} \label{subsec:intro-main-results}
We now discuss the main contributions of our paper in more detail and precision.
Our results concern torus actions, so we specialize the discussion of the preceding section and consider a $d$-dimensional
complex torus~$T = (\CC^\times)^d$ as the acting group~$G$.
The group law is just pointwise multiplication, i.e.,
$(t_1,\dots,t_d) \cdot (s_1,\dots,s_d) = (t_1s_1,\dots,t_ds_d)$.

Any linear action of a torus can be described by an integer matrix $M \in \Mat_{d,n}(\ZZ)$ called the \emph{weight matrix}
(where $\Mat_{d,n}(\ZZ)$ denotes the space of $d \times n$ integer matrices).
The representation~$\rho_M \colon T \rightarrow \GL_n(\CC)$ corresponding to a weight matrix $M=(m_{ij})$ looks as follows:
\begin{align}\label{eq:toract}
  \rho_M(t) = \begin{psmallmatrix} \prod_{i=1}^d t_i^{m_{i1}} & & \\ & \ddots & \\ & & \prod_{i=1}^d t_i^{m_{in}} \end{psmallmatrix}
\end{align}
Thus any torus action can be viewed as a scaling action, where each coordinate is scaled separately according to a Laurent monomial.%
\footnote{
We can also describe this action as follows:
Identify $v\in\CC^n$ with a Laurent polynomial $\sum_{j=1}^n v_j \, z_1^{m_{1j}} \cdots z_d^{m_{dj}}$;
then the action of $T$ corresponds precisely to rescaling the variables $z_1,\dots,z_d$~\cite{gurvits2004combinatorial}.}
The weight matrix (up to reordering of columns) determines the representation. 
Despite the simple description of commutative torus actions, they as well capture fundamental notions, and the associated orbits can be quite complex.
One example is the matrix scaling problem, where the orbits capture weights of perfect matchings (see Problem~\ref{prblm:MatchW}).

In this paper, we will assume that a torus action is given by specifying the weight matrix.
Thus the bit-length of the entries of the weight matrix are included in the input size of the problems.
Moreover, we will allow complex number inputs.
These can be described up to finite precision by elements in the field of Gaussian rationals~$\QQ(i) = \{ s + i t \;|\; s,t \in \QQ \}$,
which will be encoded in the standard way; see, e.g., \cite{GCTV}.%
\footnote{In fact, our results hold more generally when the elements in $\QQ(i)$ are given in a `floating point' format, namely in the form $(s + it) 2^p$, with $s,t\in\QQ$ and $p\in\ZZ$ encoded in binary in the standard way. The same is true for input of the form~$2^p$, with $p\in\QQ$ encoded in binary. See Remark~\ref{rem:float}.\label{foot:float}}
The following theorem captures the main results of our paper.

\begin{theorem}\label{thm:main}
Given as input a weight matrix $M \in\Mat_{d,n}(\ZZ)$ as well as vectors $v,w \in \QQ(i)^n$, denote by $b$ the maximal bit-length of the entries of $v,w$, and $M$.
Then we can in time $\poly(d,n,b)$:
\begin{enumerate}
\item decide whether $O_v = O_w$;
\item decide whether $\overline{O_v} \cap \overline{O_w} \ne \emptyset$;
\item decide whether $w \in \overline{O_v}$.
\end{enumerate}
In other words, for rational representations of tori, there are polynomial time algorithms for orbit equality, orbit closure intersection, and orbit closure containment.
\end{theorem}

We note that the null cone membership problem mentioned earlier, namely Problems~\ref{prob:main}~(2)/(3) when the input vector~$w$ is the $0$ vector,
was known to have a polynomial time algorithm by a simple reduction to linear programming.%
\footnote{
Namely, a vector $v$ is in the null cone if and only if the convex hull of the weights corresponding to the nonzero coordinates of~$v$ does not contain the origin.}
There is no known way of doing the same for the orbit problems above, and indeed our theorem above takes an alternative route.

While one might hope for efficient algorithms for Problems~\ref{prob:main}~(1) and~(2) in much more general situations than for tori (for general reductive group actions), our efficient algorithm for orbit closure containment is in stark contrast to the known NP-hardness of the general orbit closure containment problem~\cite{BILPS:20}.
Our work points to a key difference: namely, for torus group actions, one can use one-parameter subgroups combined with linear programming techniques to reduce orbit closure containment to orbit equality, while this is impossible in this form for general actions.
See Section~\ref{sec:occ} for more details.

A common core underlying all our results is an efficient algorithm for computing invariant Laurent polynomials for torus actions.
The key idea is the following.
Invariant polynomials for torus actions can be quite complicated.
However, suppose that we restrict to vectors of some fixed support, i.e., ``nonzero pattern'' of the coordinates.
This restriction is without loss of generality, since two vectors can only be in the same orbit when their supports coincide.
However, it allows us to study a richer class of functions, namely \emph{Laurent polynomials} instead of ordinary polynomials.
Allowing for negative exponents makes an important difference:
while polynomial invariants naturally form a semigroup, invariant Laurent polynomials form a \emph{lattice}, isomorphic to the integral vectors in the kernel of the weight matrix.
Lattices are much better behaved than semigroups, for example they have small bases which can be found efficiently.

Before describing our results, let us define invariant Laurent polynomials more precisely.
For a representation $\rho\colon G \rightarrow \GL(V)$ of a group $G$, we have an action of~$G$ on the polynomial ring~$\CC[V]$
defined by $(g \cdot f) (v) := f(\rho(g)^{-1}v)$.
When $V=\CC^n$, we can identify $\CC[V]=\CC[x_1,\dots,x_n]$ with the polynomial ring in~$n$ variables.
Now consider the set of vectors with nonzero coordinates in~$S\subseteq[n]$:
\begin{align*}
  X_S = \{ v \in \CC^n \;|\; v_j \neq 0 \text{ if and only if } j \in S \}.
\end{align*}
The Laurent polynomials in the variables~$x_j$ for~$j \in S$ form the natural class of functions on~$X_S$ (since we can always divide by the nonzero coordinates).
Accordingly, we will denote their collection by~$\CC[X_S]$.\footnote{In the language of algebraic geometry, these are the ``regular'' functions on~$X_S$.}
Now, for a torus action of the form~\eqref{eq:toract}, the group $T$ acts on any monomial $x^c=x_1^{c_1} \cdots x_n^{c_n}$ by a simple rescaling.
Accordingly, we also have an action of $T$ on the algebra of Laurent polynomials~$\CC[X_S]$.
A Laurent polynomial $f$ is called \emph{invariant} if $g \cdot f = f$ for all $g \in G$.
Clearly, if $f$ is invariant, then so are all the Laurent monomials occuring in~$f$.
The collection of all invariant Laurent polynomials forms the subalgebra $\CC[X_S]^G$ of \emph{invariant Laurent polynomials}.
A collection of invariant Laurent polynomials~$f_1,\dots,f_r$ is called
a \emph{system of generating invariant Laurent polynomials in the variables $\{x_j\}_{j\in S}$}
if they generate~$\CC[X_S]^G$ as an algebra.
For torus actions, these can always be taken to be Laurent \emph{monomials}, in which case we call them a \emph{system of generating invariant Laurent monomials}.
We can then state our key result:

\begin{theorem}\label{th:laurent}
Let $M \in\Mat_{d,n}(\ZZ)$ define an $n$-dimensional representation of $T = (\CC^\times)^d$, and let~$S \subseteq [n]$.
Assume that the bit-lengths of the entries of $M$ are bounded by $b$.
Then, in $\poly(d,n,b)$-time, we can construct an arithmetic circuit with division $\mathcal{C}$ whose output gates compute a system of generating invariant Laurent monomials $f_1,\dots,f_r$ in the variables $\{x_j\}_{j\in S}$, where $r \leq n$.
\end{theorem}

Here we recall the notion of an \emph{arithmetic circuit with division}, which is a directed acyclic graph as follows.
Every node of indegree zero is called an input gate and is labeled by either a variable or a rational (complex) number.
Nodes of indegree one and outdegree one are labeled by $^{-1}$ and are called divison gates.
Nodes of indegree two and outdegree one and is labeled either $+$ or $\times$; in the first case it is a sum gate and in the second a product gate.
The only other nodes allowed are output gates which have indegree one and outdegree zero. Given an arithmetic circuit with division,
it computes a rational function at each output node in the obvious way. The bit size of such an arithmetic circuit is
the total number of nodes plus the total bit-length of the specification of all rational numbers computed in {\em all} of its gates.
The notion of {\em (division free) arithmetic circuits} is obtained by disallowing division gates.
They compute polynomials in the obvious way.

We emphasize that the number of generators produced by Theorem~\ref{th:laurent} is at most~$n$ (in particular, independent of the bit-length~$b$), in stark contrast to the situation for monomial invariants.
Moreover, the bit-length of~$\mathcal{C}$ is polynomially bounded.

As a consequence of Theorem~\ref{th:laurent}, we are also able to construct arithmetic circuits that compute a generating set of \emph{rational invariants}.
For a representation $\rho\colon G \rightarrow \GL(V)$, the action of $G$ on the polynomial ring $\CC[V]$ always extends to an action on its field of rational functions, the rational functions $\CC(V)$.
A rational function $f \in \CC(V)$ is called \emph{invariant} if $g \cdot f = f$ for all $g \in G$.
The collection of all rational invariants forms the sub-field $\CC(V)^G$ of \emph{rational invariants}.
A collection of rational invariants~$f_1,\dots,f_r \in \CC(V)$ is called a \emph{system of generating rational invariants} if they generate $\CC(V)^G$ as a field extension of~$\CC$.
Note that any invariant Laurent polynomial is a rational invariant, but the converse is not necessarily true.
Nevertheless:

\begin{corollary}\label{cor:1.3}
Let $M \in\Mat_{d,n}(\ZZ)$ define an $n$-dimensional representation of $T = (\CC^\times)^d$. Assume that the bit-lengths of the entries of $M$ are bounded by $b$.  Then, in $\poly(d,n,b)$-time, we can construct an arithmetic circuit with division $\mathcal{C}$ whose output gates compute a system of generating rational invariants $f_1,\dots,f_r \in \CC(x_1,\dots,x_n)^T$, where $r \leq n$.
\end{corollary}

This result is in distinct contrast to the impossibility of finding succinct circuits for generating \emph{polynomial} invariants under natural complexity assumptions~\cite{GIMOWW}.

Furthermore, we can complement Theorem~\ref{thm:main} in the following way:
if two orbit closures do not intersect, $\overline{O_v} \cap \overline{O_w} \ne \emptyset$,
then we can construct in polynomial time an arithmetic circuit computing a separating invariant monomial that can serve as a ``witness'' of the non-intersection.

\begin{corollary}\label{cor:sep-invar}
Let $M \in\Mat_{d,n}(\ZZ)$ define an $n$-dimensional representation of $T = (\CC^\times)^d$.
Let $v,w \in \QQ(i)$ be such that $\overline{O_v} \cap \overline{O_w} = \emptyset$.
Assume the bit-lengths of the entries of $v,w$ and $M$ are bounded by~$b$.
Then, in $\poly(d,n,b)$-time, we can construct an arithmetic circuit of bit-length $\poly(d,n,b)$,
which computes an invariant monomial $f$ such that $f(v) \neq f(w)$.
\end{corollary}

So far, we have discussed orbit problems for complex tori $T = (\CC^\times)^d$.
It is interesting to ask to which extent our results hold for ``compact'' tori, which are groups of the form~$K = (\Ss^1)^d$, where $\Ss^1 = \{z \in \CC^\times \ | \ |z| = 1\}$.%
\footnote{Note that $K$ is indeed compact, and a subgroup of $T$. Moreover, any commutative compact connected Lie group is of this form.}
Besides the fundamental algorithmic interest in this setting, such group actions are important in several areas.
For example, the time evolution of periodic systems in Hamiltonian mechanics are naturally given by $S^1$-actions, and important symmetries in classical and quantum physics are given by compact group actions.

In fact, the results discussed so far can also be used to give an efficient solution for orbit problems for compact tori.
Any (continuous) finite-dimensional representation of $(\Ss^1)^d$ extends to a representation of $(\CC^\times)^d$, so representations are specified as before by a weight matrix $M \in\Mat_{d,n}(\ZZ)$.
Moreover, the compactness implies that orbits $O_{K,v} = \{kv \ |\ k \in K\}$ are closed and so all three problems mentioned in Problem~\ref{prob:main} coincide.
Therefore, the following corollary solves all three problems for compact tori:

\begin{corollary}\label{cor:main}
Let the weight matrix $M \in\Mat_{d,n}(\ZZ)$ define an $n$-dimensional representation of  $T=(\CC^\times)^d$ and put $K = (\Ss^1)^d$.
Further, let $v,w \in \QQ(i)^n$ and assume that the bit-lengths of the entries of $v,w$ and $M$ are bounded by $b$.
Then, in $\poly(d,n,b)$-time, we can decide if $O_{K,v} = O_{K,w}$.
\end{corollary}

To give additional context to this result, we briefly mention some recent results achieving polynomial time algorithms for orbit closure intersection of \emph{specific} group actions.
For the left-right action (of $\SL_n \times \SL_n$ on $m$-tuples of $n \times n$ matrices), one approach to solving the orbit closure intersection problem is to (approximately) reduce to the orbit equality problem for the maximal compact subgroup (which happens to be ${\rm SU}(n) \times {\rm SU}(n)$, where ${\rm SU}(n)$ denotes the group of $n \times n$ unitary matrices with determinant $1$), see\cite{AZGLOW}.
This was achieved by using a geodesic convex optimization algorithm.
Given the recent advances in this area (see, e.g., \cite{BFGOWW2} and references therein), it is natural to ask if a similar approach could be useful for general reductive group actions.
For torus actions, interestingly, we can also go in the other direction.
Namely, our result for the orbit equality problem for the maximal compact subgroup, Corollary~\ref{cor:main}, is derived from our main result for complex tori, i.e., Theorem~\ref{thm:main}.
More generally, we observe that for arbitrary reductive group actions, the orbit equality problem for the maximal compact subgroup is always equivalent to an orbit closure intersection (or equality) problem for a related action of the larger group, see Theorem~\ref{thm:compacttogeneral} for a precise statement.


The results in this paper warrant the investigation of several interesting directions that we leave for future work, some of which we will discuss in Section~\ref{sec:future}.

\subsection{Further motivation and algorithmic applications}\label{subsec:applications}
As we saw above, orbit problems are related to a great number of applications.
Despite significant progress, for general reductive group actions it is still an open problem to design fast algorithms for these problems.
Our results fully resolve the situation in the case of torus actions and also show how to overcome barriers that had previously been pointed out in the literature~\cite{ikenmeyer2017vanishing,GIMOWW}.
Apart from its fundamental complexity theoretic interest, there are also several algorithmic applications where torus actions arise naturally.
Here we discuss in more detail some concrete applications to combinatorial optimization and to dynamical systems,
which were already mentioned briefly at the beginning of the introduction.

We first explain a link to combinatorial optimization.
Consider edge weights $w$ for the complete bipartite graph on $2n$ labeled vertices ($n$ on each side):
the weight $w(e)$ of an edge~$e$ is assumed to be a rational number, encoded in binary.
We define the weight $w(M)$ of a perfect matching~$M$ of~$G$ as the sum of the weights of the edges occurring in~$M$.

\begin{problem}\label{prblm:MatchW}
Given edge weights $w$ and $w'$ as above, decide whether they assign the same weight to \emph{every} perfect matching $M$ of $G$.
\end{problem}

Perhaps surprisingly, this problem can be reformulated as an orbit intersection problem for
a torus action (see below).
Therefore, Theorem~\ref{thm:main} implies that Problem~\ref{prblm:MatchW} can be solved in polynomial time.
This insight seems far from being obvious!

The relevant torus action here results from from matrix scaling, which has been widely studied and has
many applications; see~\cite{Sink} and~\cite{cohen2017matrix} for more recent developments.
Consider $\ST_n:= \{ (t_1,\ldots,t_n) \in \CC^\times \mid t_1\cdots t_n =1\}$, which is isomorphic to the algebraic torus $(\CC^{\times})^{n-1}$.
We let $\ST_n\times\ST_n$ act on $\Mat_n(\CC)$ by left-right multipliation as follows:
\begin{equation}\label{eq:ma-scal}
 ((t_1,\ldots,t_n),(s_1,\ldots,s_n)) \cdot (v_{ij}) := (t_i v_{ij} s_j)_{ij} .
\end{equation}
Moreover, we shall identify the edge weights~$w_{ij}$, where $i,j \in [n]$, with the matrix $v_w=(2^{w_{ij}}) \in\Mat_n(\CC)$.%
\footnote{As explained in footnote~\ref{foot:float}, our results also hold for input of this form, where the $w_{ij}$ are specified in binary.}
Then one can show that the answer to Problem~\ref{prblm:MatchW} is affirmative if and only if the orbit closures of~$v_w$ and~$v_{w'}$ in~$\Mat_n(\CC)$ intersect.
This follows from Mumford's theorem mentioned earlier, along with the fact that the invariant polynomials for this action are generated by the perfect matchings, namely the monomials~$f_\pi = x_{1,\pi(1)} \cdots x_{n,\pi(n)}$ where $\pi\in S_n$ ranges over the permutations~\cite[Theorem~3]{leep1999marriage}.
Indeed, multiplying entries of $v_{w}$ is the same as summing the corresponding edge weights in the exponent, hence~$f_\pi(v_w) = 2^{w(M)}$, where $M$ is the perfect matching defined by the permutation $\pi$.


We briefly comment on the 3-dimensional generalization of this action.
Here, $\ST_n \times \ST_n \times \ST_n$ acts on 3-tensors in $\CC^n \otimes \CC^n \otimes \CC^n$ in the natural way:
\begin{align*}
  ((t_1,\dots,t_n),(s_1,\dots,s_n),(u_1,\dots,u_n)) \cdot (v_{ijk}) = (t_i s_j u_k v_{ijk})_{ijk}.
\end{align*}
In this case, any system of generating polynomial invariants must include the (maximum) 3-dimensional matching monomials $f_{\pi,\tau} = x_{1,\pi(1),\tau(1)} \cdots x_{1,\pi(n),\tau(n)}$ for~$\pi,\tau\in S_n$, which led to the barrier result for torus actions in~\cite{GIMOWW}.
Of course, in this case there are additional generating invariants, see, e.g., \cite{linial2014vertices}.
Our results show that the corresponding orbit problems can nevertheless be solved in polynomial time!
Moreover, it is possible to efficiently exhibit \emph{separating} polynomial invariants (whenever they exists) as well as to construct systems of generating invariant \emph{Laurent polynomial} or \emph{rational} invariants.

Our second example concerns a connection to dynamical systems.
Consider a (massless) cue ball on a billiard table (assumed to be square to simplify the discussion).
We can ask:

\begin{problem}\label{prblm:billiard}
If we hit the cue ball at a given angle, will its trajectory end up in a pocket?
\end{problem}

It is well-known, and easy to see, that one can map trajectories on an ordinary billiard with reflecting boundaries to a billiard of twice the size with periodic boundaries, say $(\RR/2\pi\ZZ)^2$.
The trajectory of the ball depends fundamentally on the angle or slope. 
If the slope is irrational, then the trajectory will be dense, so the answer to Problem~\ref{prblm:billiard} is trivially yes.
Otherwise, the trajectory will be periodic and the problem is nontrivial.
We can model it as an orbit problem as follows.
Let the compact torus $\Ss^1$ act on $\CC^2$ by
\begin{align*}
  t \cdot (x,y) := (t^p x, t^q y),
\end{align*}
where $s=\frac qp$ is the slope by which we hit the ball.
We can identify points $(\theta,\nu)$ on the periodic billiard with points $(e^{i\theta},e^{i\nu}) \in \CC^2$.
In this way, Problem~\ref{prblm:billiard} reduces to a constant number of orbit equality problems for this action (one for each pocket).
While the problem is certainly easy to solve by a variety of methods, one can ask analogous questions for billiards in $n>2$~dimensions and by allowing a $d$-dimensional hyperplane worth of allowed cue directions.
Such generalizations similarly correspond to orbit problems for compact tori~$(\Ss^1)^d$ on some~$\CC^n$, and they can all be solved in polynomial time by using Corollary~\ref{cor:main}.


\subsection{Organization of the paper}\label{subsec:organization}
In Section~\ref{sec:inv}, we give an introduction to basic results in invariant theory that we will need to establish our results.
In Section~\ref{sec:rep-tori}, we focus on tori, their representations, and their invariants.
In particular, we will show that the faces of a natural convex polyhedral ``Newton cone'' are in one-to-one correspondence with the orbits in an orbit closure, which will be an important ingredient later on.

In Section~\ref{sec:gen-inv}, we discuss the definition and computation of suitable \emph{rational} invariants.
As mentioned above, our key result is that for fixed support, a small generating set of invariant Laurent monomials can be computed efficiently.
This result, which is Theorem~\ref{th:laurent}, is at the heart of our algorithms, and also of independent interest.
We achieve this using Smith normal forms.
As an easy consequence, this also implies that we can efficiently compute a small generating set of rational invariants for a given representation, that is, Corollary~\ref{cor:1.3}.

In Section~\ref{sec:oe}, we explain how to use the results of the preceding section to solve the orbit equality problem in polynomial time.
This establishes part~(1) of Theorem~\ref{thm:main}.
Here we rely on known results for testing if a given Laurent monomial (of possibly exponential degree) evaluates to the same value on two given vectors, and we present a brief sketch for completeness.

In Sections~\ref{sec:oci-to-oe} and \ref{sec:occ}, we show how to solve the orbit closure intersection and containment problems by reducing them to orbit equality.
This establishes parts~(2) and (3) of Theorem~\ref{thm:main}.
Here we use the polyhedral description of the structure of orbit closures as furnished by the Newton cone.
Furthermore, we show that given two points whose orbit closures do not intersect, we can efficiently construct a separating monomial invariant as a ``witness''.
This proves Corollary~\ref{cor:sep-invar}.

In Section~\ref{sec:compact}, we show how to solve the orbit equality problem for compact tori.
This establishes Corollary~\ref{cor:main}.
We also give, for general reductive groups~$G$, a reduction from orbit equality for a maximally compact subgroup~$K\subseteq G$ to orbit equality and orbit closure intersection for~$G$.

In Section~\ref{sec:future}, we summarize our results and discuss some important open problems and future directions.

\medskip\noindent\textbf{Conventions.}
In this paper, sometimes we work with monomials and sometimes with Laurent monomials.
Unless we use the prefix ``Laurent'', by a monomial, we mean $\prod_j x_j^{c_j}$ where $c_j \in \NN$, i.e., all exponents are non-negative.
Whenever exponents are allowed to be negative, we will be careful to specify that it is a Laurent monomial.

\section{Preliminaries of invariant theory}\label{sec:inv}
We will briefly recall the main results in invariant theory that are relevant for us (see~\cite{kraft,dolgachev:03,DK,Mumford-book} for details).
We will take our ground field to be $\CC$, the field of complex numbers, for simplicity.
However, much of this theory works for any algebraically closed field.
For a (finite-dimensional) vector space $V$, we denote by $\CC[V]$ the ring of polynomial functions on~$V$.
For our purposes, if $V$ is the standard vector space $\CC^n$, then $\CC[V] = \CC[x_1,\dots,x_n]$, the polynomial ring in $n$ variables, where $x_i$ is to be interpreted as the $i^{th}$ coordinate function.

Let $G$ be an algebraic group, i.e., it has the structure of an algebraic variety (not necessarily irreducible) such that the multiplication map $m\colon G \times G \rightarrow G$ and the inverse map $\iota\colon G \rightarrow G$ are morphisms of varieties.%
\footnote{A morphism of varieties simply means that in local coordinates the map is given by ratios of polynomials. For concreteness, the reader may simply think of an algebraic group as a matrix group, i.e., a subgroup of $\GL_n(\CC)$ that is described as the zero locus of a collection of polynomials.}
A morphism of algebraic groups $\rho: G \rightarrow \GL(V)$ is called a rational representation of $G$.%
\footnote{One can interpret this action as the action of the subgroup $\rho(G) \subseteq \GL(V)$ on $V$ by matrix-vector multiplication, where $\rho(G)$ is parametrized algebraically by an algebraic group $G$.}
We write $gv$ or $g \cdot v$ for $\rho(g) v$.
For a point $v \in V$, its orbit $O_v$ (or $O_{G,v}$ when the group is not clear from context) is the set of all points that can be reached from $v$ by the action of an element of the group, i.e.,
\begin{align*}
  O_v := \{gv \ |\ g \in G\}.
\end{align*}

We denote by $\overline{O_v}$ the closure of the orbit $O_v$.
The closure is to be taken either with respect to the Euclidean topology or the Zariski topology.
Indeed, the closures in both topologies coincide, a well-known fact that relies on a fundamental result in algebraic geometry due to Chevalley (see~\cite[I.\S10]{mumford:88}).
A polynomial function $f \in \CC[V]$ is called \emph{invariant} if it is oblivious to the group action, i.e., $f(gv) = f(v)$ for all $g \in G$, $v \in V$.
The collection of all invariant polynomials forms a subring
\begin{align*}
  \CC[V]^G := \{f \in \CC[V]\ |\ \forall\ g\in G, v \in V\ f(gv) = f(v) \}.
\end{align*}
One key observation is that invariant functions are constant along orbits and hence constant along orbit closures as well.
Hence, if the orbit closures of two points intersect, then they cannot be distinguished by an invariant function.
The converse was proved by Mumford for a special class of groups called reductive groups~\cite{Mumford-book} (see also~\cite[Corollary~2.3.8]{DK}).
An algebraic group $G$ is called {\em reductive} if every rational representation is a direct sum of irreducible representations,
wherein a representation is called irreducible if it has no non-trivial subrepresentations.
Examples of reductive groups include $\SL_n, \GL_n, {\rm Sp}_n, {\rm O}_n$, finite groups, and most importantly for us, tori (which we define formally in the next section),
as well as direct products thereof.\footnote{The group $B_n$ of upper triangular $n \times n$ invertible matrices is a typical example of a group that is not reductive.}
Reductive groups have played a central role for a number of mathematical fields for over a century. A particularly important result in the invariant theory of reductive groups is that invariant rings are finitely generated \cite{Hilb1, Hilb2,weyl:39}.

To state Mumford's result in the generality we need, we will define rational actions on varieties (a notion that naturally generalizes rational representations).
Let $X$ be an algebraic variety and let $\CC[X]$ denote the ring of regular functions on~$X$.
A rational action of an algebraic group $G$ on~$X$ is a morphism of varieties $G \times X \rightarrow X, (g,x) \mapsto g\cdot x$ satisfying $g \cdot (g' \cdot x) = (gg') \cdot x$ and $e \cdot x = x$ for all $x \in X$, $g,g' \in G$.
As in the vector space case, we denote the orbit of a vector~$v\in X$ by $O_v$.

\begin{theorem}[Mumford, \cite{Mumford-book}]\label{thm:Mumford}
Let $G$ be a reductive group. Let $X$ be an algebraic variety and suppose we have a rational action of $G$ on $X$. For $v,w \in X$ we have
$\overline{O_v} \cap \overline{O_w} = \emptyset$ if and only if there exists $f \in \CC[X]^G$ such that $f(v) \neq f(w)$.
\end{theorem}

Another well-known important structural result states that every orbit closure $\overline{O_v}$ contains a {\em unique} closed orbit.

\begin{theorem}\label{th:uniqueclosedorbit}
Let $\rho\colon G \rightarrow \GL(V)$ be a rational representation of a reductive group $G$. Then:
\begin{enumerate}
\item For any $v\in V$, the orbit closure $\overline{O_v}$ contains a unique closed orbit, that we denote by $O_{\widetilde{v}}$.
\item\label{it:oci red} If $v,w \in V$, then
\[
\overline{O_v} \cap \overline{O_w} \ne \emptyset \Longleftrightarrow O_{\widetilde{v}} = O_{\widetilde{w}}.
\]
\end{enumerate}
\end{theorem}
\begin{proof}
(1) The first assertion is \cite[Theorem~2.3.6]{DK}.

(2) For the second assertion, if the orbit closures $\overline{O_v}$ and $\overline{O_w}$ are disjoint, then so are the orbits~$O_{\widetilde{v}}$ and~$O_{\widetilde{w}}$, which therefore must be different.
Conversely, suppose $O_{\widetilde{v}} \neq O_{\widetilde{w}}$.
Since these orbits are closed, by Theorem~\ref{thm:Mumford},
there is an invariant $f \in \CC[V]^G$ such that $f(\widetilde{v}) \neq f(\widetilde{w})$.
By continuity, $f(v) = f(\widetilde{v}) \neq f(\widetilde{w}) = f(w)$, which implies $\overline{O_v} \cap \overline{O_w} = \emptyset$ by another application of Theorem~\ref{thm:Mumford}.
\end{proof}

Part(2) of this theorem shows that the orbit closure intersection problem can be reduced to the orbit equality problem,
provided we can compute the unique closed orbit $O_{\widetilde{v}}$ contained in $\overline{O_v}$.
We will see in Section~\ref{sec:oci-to-oe} that if the group~$G$ is a torus, this can be achieved in polynomial time.

Another key result in understanding orbit closures is the Hilbert--Mumford criterion.
A {\em one-parameter subgroup} of $G$ is a morphism of algebraic groups $\sigma \colon \CC^\times \rightarrow G$.
For a representation of~$G$ on a vector space~$V$, we say that a subset $S \subseteq V$ is $G$-stable if $g \cdot s \in S$ for all $g \in G$, $s \in S$.


\begin{theorem}[Hilbert--Mumford criterion, \cite{Hilb2,Mumford-book}]\label{th:HM-crit}
Let $\rho \colon G \rightarrow \GL(V)$ be a rational representation of a reductive group $G$.
Suppose $S \subseteq V$ is a $G$-stable closed subvariety of $V$ and let~$v \in V$ such that $\overline{O_v} \cap S \neq \emptyset$.
Then there exists a one-parameter subgroup $\sigma \colon \CC^\times \rightarrow G$ such that $\lim_{\epsilon \to 0} \sigma(\epsilon) \cdot v \in S$.
\end{theorem}

A particular use of the above theorem is to take $S = \{0\}$ or $S = O_{\tilde v}$.
When $G$ is a torus, the set of one-parameter subgroups has the structure of a $\ZZ$-lattice. We will discuss this further in the next section.

We end this section by introducing a key notion in invariant theory called the \emph{null cone}, whose significance will become clear in later sections.
For a collection $F$ of polynomials in $\CC[V]$, we denote by~$\mathbb{V}(F)$ their common zero locus in $V$.

\begin{definition}[Null cone]\label{def:null}
Let $\rho\colon G \rightarrow \GL(V)$ be a rational representation of a reductive group~$G$.
Then the \emph{null cone} is defined as
\[
\mathcal{N}_G(V) := \mathcal{N}(\rho) := \{v \in V\ |\ 0 \in \overline{O_v}\}.
\]
It can also be defined as the common zero locus of all invariant polynomials without constant part:
\[
\mathcal{N}_G(V) := \mathcal{N}(\rho) := \mathbb{V}(\bigcup_{d > 0} \CC[V]^G_d),
\]
where $\CC[V]^G_d$ denotes the space of invariant polynomials that are homogeneous of degree $d$.
The equivalence of the two definitions of the null cone follows from Theorem~\ref{thm:Mumford}.
\end{definition}

\section{Invariants and orbit closures of torus actions}\label{sec:rep-tori}
Invariant theory for general reductive groups can get very complicated.
However, for representations of tori, that is, \emph{commutative} connected reductive groups, a lot of the theory can be viewed as a combination of linear algebra and the study of convex polytopes.
We will collect important results regarding torus actions in this section and refer the reader to~\cite{Wehlau93, DK} for more details.
All the results in this section are already known or can be deduced from the existing literature, and we provide proof sketches for completeness.
Note that tori are reductive groups, so the results of the previous section hold in this setting.

We will first briefly recall torus actions and the notions of characters/weights, one-parameter subgroups and how weight matrices define a representation. Then, we give a linear algebraic description of invariant rings by determining the monomials that are invariant. Then, we describe a polyhedral perspective on orbits. In particular given a point $v$ in the vector space of the representation, we define a polyhedral cone, called the Newton cone. The Newton cone can be used to determine whether $v$ is in the null cone and moreover we give a correspondence between the faces of the Newton cone to orbits in the orbit closure of $v$, which is crucial in understanding the orbit closure containment problem.

For this entire section, fix a torus $T = (\CC^\times)^d$.\footnote{Any commutative connected reductive group is isomorphic to some $(\CC^\times)^d$. Important examples include ${\rm T}_d$, the group of diagonal $d \times d$ invertible matrices and its subgroup ${\rm ST}_d$ consisting of diagonal matrices with determinant $1$.}

\subsection{Representations and invariants}\label{subsec:rep inv tori}
As described in Section~\ref{subsec:intro-main-results}, any representation of a torus~$T$ is a ``scaling'' action
(after identifying $V$ with $\CC^n$ by an appropriate choice of basis).
Namely, each coordinate of $v \in \CC^n$ is multiplied by some (Laurent) monomial $\prod_{i=1}^d t_i^{\lambda_i}$ for integers~$\lambda_i \in \ZZ$.
These monomials (succinctly described by the so-called weight matrix, see below) together specify the representation. 
We now make this more precise.

A 1-dimensional (rational) representation is called a \emph{character} or a \emph{weight}.
Let $\mathcal{X}(T)$ denote the set of weights of $T$, which forms a group where the binary operation is  (pointwise) multiplication of functions.
To each $\lambda = (\lambda_1,\dots,\lambda_d) \in \ZZ^d$, we associate a weight, also denoted $\lambda$ by slight abuse of notation, namely
\[ \lambda\colon T \rightarrow \CC^\times, \quad \lambda(t) = \prod_{i=1}^d t_i^{\lambda_i}, \]
which gives an identification of abelian groups $\ZZ^d \cong \mathcal{X}(T)$.

Let $\rho \colon T \rightarrow \GL(V)$ be a (rational) representation of~$T$ where $V$ is an $n$-dimensional vector space.
We can choose a basis of $V$ consisting of weight vectors, wherein a vector~$v \in V$ is called a \emph{weight vector}
of weight~$\lambda \in \mathcal{X}(T)$ if $t \cdot v = \lambda(t) v$ for all $t \in T$.
Once we have chosen a weight basis, using the identification~$\mathcal{X}(T) \cong \ZZ^d$, the corresponding $n$ weights
can be collected into a $d \times n$ matrix with integer entries, which we call the \emph{weight matrix} of the representation.
Up to permutation of the columns, it is independent of the choice of weight basis, and it classifies the representation up to isomorphism.
Concretely, a matrix $M = (m_{ij}) \in \Mat_{d,n}(\ZZ)$ describes the representation $\rho_M \colon T \rightarrow \GL_n(\CC)$
defined in \eqref{eq:toract}.
That is, for $t = (t_1,\dots,t_d)$ and $v = (v_1,\dots,v_n) \in \CC^n$, we have
\[
 t \cdot v = \rho_M(t) v = \left( \left(\prod_{i=1}^d t_i^{m_{i1}}\right) v_1, \left(\prod_{i=1}^d t_i^{m_{i2}}\right) v_2, \dots, \left(\prod_{i=1}^d t_i^{m_{in}}\right) v_n \right).
\]
The matrix $M$ is the weight matrix for this action.
The $j^{th}$ standard basis vector~$e_j$ is a weight vector of weight~$m^{(j)} = (m_{1j}, m_{2j},\dots, m_{dj}) \in \ZZ^d = \mathcal{X}(T)$.
Note that $m^{(j)}$ is the $j^{th}$ column vector of $M$.

For the rest of this section, we fix an $n$-dimensional representation $\rho_M \colon T \rightarrow \GL_n(\CC)$ of the torus~$T=(\CC^\times)^d$
given by a weight matrix $M \in \Mat_{d,n}(\ZZ)$ with columns $m^{(j)}$ for $j\in[n]$.
The following well-known result describes the invariant ring of this action (see, e.g., \cite[Section~3]{DM-exp}):

\begin{proposition}\label{pro:3.1}~
\begin{enumerate}
\item\label{it:semigroup} Let $c\in\NN^n$. A monomial $x^c=\prod_j x_j^{c_j}$ is invariant if and only if $\sum_j c_j m^{(j)} = 0$;
\item The invariant ring $\CC[x_1,\dots,x_n]^T$ is spanned as a vector space by the invariant monomials.
\end{enumerate}
\end{proposition}
\begin{proof}
For the action~$\rho$ of $G$ on $V$, there is a natural induced action of $G$ on the ring of polynomial functions~$\CC[V]$ defined by the formula $g \cdot f (v) := f(\rho(g)^{-1}v)$.
Applying this for the action $\rho_M$, we get an induced action of $T$ on $\CC[x_1,\dots,x_n]$.
It is easy to compute this action:
for a monomial~$x^c$ and $t \in T$, we have $t \cdot x^c = \lambda(t)^{-1} \, x^c$, where $\lambda \in \mathcal{X}(T)$ is the character corresponding to~$\sum_j c_j m^{(j)} \in \ZZ^d$.
It follows that the monomials which are invariant are precisely the ones for which $\sum_j c_j m^{(j)} = 0$, the trivial character, proving the first part.
The second part follows from the observation that a polynomial is invariant if and only if each monomial that occurs in it is invariant.
\end{proof}

Part~(\ref{it:semigroup}) of Proposition~\ref{pro:3.1} shows that the invariant monomials are in one-to-one correspondence with the nonnegative integer vectors in the kernel of the weight matrix.
Accordingly, they form a semigroup.
In general, such semigroups can have a large number of generators, which explains the difficulty of using polynomial invariants~\cite{durand-et-al:02}.
Our key idea to obtain efficient algorithms will be to instead consider invariant Laurent monomials, which form a lattice rather than a semigroup.
We will return to this in Section~\ref{sec:gen-inv}.

In turns out that the weights lead to a strong link to convex polyhedral geometry, which in turn characterizes the orbits in an orbit closure.
For this, we make the following definitions.
The \textit{support} of a vector $v\in \CC^n$ is defined as
\[
  \supp(v) := \{j\in [n]\mid v_j\neq 0\}.
\]
Let us record some of the properties of the support. By dimension (of an orbit, orbit closure, algebraic group, etc), we mean the dimension of the underlying variety.

\begin{lemma} \label{lem:supp-dim}
For $v,w \in \CC^n$ we have:
\begin{enumerate}
\item If $O_v = O_w$, then $\supp(v) = \supp(w)$.
\item If $\supp(v) = \supp(w)$, then $\dim O_v = \dim O_w$.
\item\label{it:supp down cl} If $w \in \overline{O_v}$, then $\supp(w) \subseteq \supp(v)$.
This inclusion is strict if and only if $w \in \overline{O_v} \setminus O_v$.
\end{enumerate}
\end{lemma}
\begin{proof}
(1) is clear, since each coordinate simply gets rescaled by a nonzero number by the group action.
For (2) we note that the stabilizer group $\stab(v)$ of $v$ only depends on $\supp(v)$.
The claim follows using $\dim O_v = d - \dim \stab(v)$.
For~(3), the inclusion of supports holds since taking limits can never increase the support.
Finally, it is known~\cite[\S8.3]{Humphreys:75} that $\overline{O_v} \setminus O_v$ is a Zariski closed subset of dimension
strictly less than $\dim O_v$. Hence $w \in \overline{O_v} \setminus O_v$ implies $\dim O_w < \dim O_v$
and therefore $\supp(w) \subsetneq \supp(v)$ by part~(2).
\end{proof}

\subsection{Newton cone and orbit closures}
We define the \emph{Newton cone}~$C(v)$ of a vector~$v\in \CC^n$ to be the rational polyhedral cone generated by
the weights corresponding to the indices in the support, that is,
\[
 C(v) := \Big\{\sum_{j\in\supp(v)} c_j m^{(j)} \mid c_j \ge 0 \Big\}\subseteq\RR^d .
\]
The {\em lineality space} of the cone $C(v)$ is defined as
$L(v) := C(v) \cap (-C(v))$. Clearly, it is the largest linear subspace contained in $C(v)$.
The cone $C(v)$ is called {\em pointed} iff $L(v)=0$.
(Compare~\cite{schrijver:86} for the structure of polyhedral cones.)

These notions are standard in geometric programming, which essentially studies optimization problems associated with torus actions, albeit often with a different representation and motivation; see, e.g.,~\cite{BLNW:20} and references therein.
The connection is particularly apparent and useful in the study of polynomial capacities which have important applications to approximate counting~\cite{linial2000deterministic,gurvits2004combinatorial}.

We will see that the Newton cone contains all the information about the orbits contained in an orbit closure.
To start, we show that membership in the null cone can be characterized as follows.
Define the \emph{essential support} of a vector $v \in V$ as
\begin{equation}\label{def:esupp}
 \esupp(v) := \{ j \in \supp(v) \mid m^{(j)} \in L(v) \}.
\end{equation}

\begin{lemma}\label{lem-esupp}
Let $k \in \supp(v)$.
We have $k \in \esupp(v)$ if and only if there exists an invariant monomial $\prod_{j \in \supp(v)} x_j^{c_j}$ with $c_j \in \ZZ_{\geq 0}$ such that $c_k > 0$.
\end{lemma}
\begin{proof}
It is easy to see that $m^{(k)} \in L(v)$ if and only if there is a non-negative integral linear combination $\smash{\sum_{j \in \supp(v)} c_j m^{(j)}} = 0$ with $c_k > 0$.
By Proposition~\ref{pro:3.1}, this is equivalent to the existence of an invariant monomial $\smash{\prod_{j \in \supp(v)} x_j^{c_j}}$ with $c_j \in \ZZ_{\geq 0}$ such that $c_k > 0$.
\end{proof}

\begin{corollary}\label{cor:null cone}
We have that $v$ is in the null cone $\mathcal{N}(\rho_M)$ if and only if $\esupp(v) = \emptyset$.
\end{corollary}

\noindent
Equivalently, $v$ is in the null cone if and only if $C(v)$ is pointed and $m^{(j)} \neq 0$ for all $j \in \supp(v)$.




In fact, much more can be said.
Let us first recall the notion of faces of polyhedral cones.
If $C(v)$ is contained in a closed halfspace $H_+$ of $\RR^d$ bounded by a linear hyperplane $H$,
then we call the intersection $F=H\cap C(v)$ a {\em face} of $C(v)$ when it is non-empty.
The cone itself is also considered a face of $C(v)$: by definition, it is the largest face of $C(v)$.
On the other hand, each face of $C(v)$ must contain the lineality space $L(v)$, which is therefore the smallest face of $C(v)$.

We will see shortly that the faces of $C(v)$ are in bijective correspondence with the orbits contained in $\overline{O_v}$.
For this, we need to introduce some more notation.
For a subset $J \subseteq \supp(v)$, we define the \emph{restriction}~$v|_{J}$ to be the vector with entries
\[
(v|_{J})_j = \begin{cases}
 v_j & \mbox{ if } j \in J,\\
 0 & \mbox{ otherwise,}
 \end{cases}
\] as its $j$-th coordinate.
Let now $F$ be a face of $C(v)$ defined by a closed half-space $H_+ = \{ y \in \RR^d \;|\; \nu \cdot y \geq 0 \}$ for some $\nu\in\RR^d$, that is,
\[
 F= \{ y\in C(V) \mid \nu \cdot y = 0 \}.
\]
Since $C(v)$ is rational, we may assume that $\nu$ has integer components.
We assign to $F$ the subset of indices
\[
 S_F := \{ j \in \supp(v) \mid m^{(j)} \in F \}
\]
and define $v_F := v|_{S_F}$.
Let us check that the orbit $O_{v_F}$ of $v_F$ is contained in $\overline{O_v}$.
The one-parameter subgroup $\sigma\colon \CC^\times \to T$ given by $\sigma(\epsilon) = (\epsilon^{\nu_1},\ldots,\epsilon^{\nu_d})$ satisfies
\begin{equation}\label{eq:SEV}
 \sigma(\epsilon) \cdot v = \rho_M(\sigma(\epsilon)) v = (\epsilon^{\nu \cdot m^{(1)}} v_1,\ldots,\epsilon^{\nu \cdot m^{(n)}} v_n) .
\end{equation}
It follows that $\lim_{\epsilon\to 0} \sigma(\epsilon) \cdot v = v_F$ and hence $v_F \in \overline{O_v}$.
The same reasoning shows that $v_{F} \in \overline{O_{v_{F'}}}$ if $F$ is a face contained in the face $F'$.

The following result is well known, see e.g., \cite[Example~1.3]{Popov-occ}, but we sketch a proof for completeness.

\begin{proposition} \label{prop:orbit-structure}
The map $F\mapsto O_{v_F}$ is a bijection between the set of faces of $C(v)$ and
the set of orbits contained in $\overline{O_v}$. Moreover, we have
$$
  F\subseteq F' \Longleftrightarrow \overline{O_{v_F}}\subseteq\overline{O_{v_{F'}}}\ .
$$
\end{proposition}

The proof of surjectivity relies on a strengthening of the Hilbert--Mumford criterion (Theorem~\ref{th:HM-crit}).
Recall this states that if we consider a closed subset $S$ that is stable under the group action and intersects the orbit closure of some point $v$,
then there is a one-parameter subgroup that will drive $v$ to a point in $S$ in the limit.
However, a subtle point is that this requires $S$ to be closed. In general, orbits are not closed, so a point $w$ could be in the orbit closure of a point $v$, but the orbit of $w$ may not be closed.
In this case, Theorem~\ref{th:HM-crit} does not apply to $S=O_w$, and indeed the orbit of~$w$ need not be reachable from~$v$ by a limit of a one-parameter subgroup.
The following theorem shows that for torus actions such a phenomenon does not happen.
This crucial fact will also prove useful for us algorithmically in Section~\ref{sec:occ}.

\begin{theorem}[\cite{kraft}, Kapitel~III.2.2]{\label{thm:hmtori}}
Let $\rho\colon T \rightarrow \GL(V)$ be a rational representation.
Suppose $v,w \in V$ are such that $w\in\overline{O_v}$.
Then there exists a one-parameter subgroup $\sigma\colon\CC^\times\rightarrow T$ such that
\[
\lim_{\epsilon\to 0}\sigma(\epsilon) \cdot v \in O_w.
\]
\end{theorem}

Before we prove Proposition~\ref{prop:orbit-structure}, we discuss a bit about the structure of one-parameter subgroups.
For each $\nu \in \ZZ^d$, we define a one-parameter subgroup of $T$, namely $\sigma\colon \CC^\times \rightarrow T$ defined
by $\sigma(\epsilon) = (\epsilon^{\nu_1},\ldots,\epsilon^{\nu_d})$.
Any one-parameter subgroup of $T$ is of this form.
This gives an identification of abelian groups $\ZZ^d \cong \mathcal{Y}(T)$, where $\mathcal{Y}(T)$ denotes the collection of all one-parameter subgroups of $T$.

We leave the proof of the following well known lemma to the reader.

\begin{lemma}\label{lem:tori-1psg}
Let $\sigma\colon \CC^\times \rightarrow T$ be a one-parameter subgroup, so $\sigma(\epsilon) = (\epsilon^{\nu_1},\ldots,\epsilon^{\nu_d})$ for some~$\nu\in\ZZ^d$, and let $v\in \CC^n$.
\begin{enumerate}
\item The limit $\lim_{t \to 0} \sigma(t) \cdot v$ exists if and only if $m^{(j)} \cdot \sigma \geq 0$ for all $j \in \supp(v)$.
\item If the limit exists, then $\lim_{t \to 0} \sigma(t) \cdot v = v|_S$, where $S = \{j \in \supp(v) \ | \ m^{(j)} \cdot \sigma =  0\}$.
\end{enumerate}
\end{lemma}

\begin{proof}[Proof of Proposition~\ref{prop:orbit-structure}]
We have already verified that $O_{v_F}$ is an orbit contained in $\overline{O_v}$, hence $F \mapsto O_{v_F}$ is well-defined as a map from the set of faces of~$C(v)$ to the set of orbits contained in~$\overline{O_v}$.
To see that it is injective, note that $F$ is the cone generated by $\supp(v_F) = S_F$.
For surjectivity, let~$O_w$ be an orbit contained in $\overline{O_v}$ and $\sigma\colon\CC^\times\rightarrow T$ be a one-parameter subgroup as in Theorem~\ref{thm:hmtori}.
There is $\nu\in\ZZ^d$ such that
$\sigma(\epsilon) = (\epsilon^{\nu_1},\ldots,\epsilon^{\nu_d})$.
By Lemma~\ref{lem:tori-1psg}, the existence of $\lim_{\epsilon\to 0}\sigma(\epsilon) \cdot v$ means that $\nu\cdot m^{(j)} \ge 0$ for all $j\in \supp(v)$.
In other words, $C(v)$ is contained in the halfspace
$\{ y\in\RR^d \mid \nu \cdot y \ge 0 \}$.
Moreover, the limit equals $v_F$, where~$F$ is the face
$F := \{ y\in C(v) \mid \nu \cdot y = 0 \}$ of $C(v)$.
Therefore, $v_F \in O_w$, hence $O_{v_F}=O_w$, and we have shown surjectivity.

In order to show the remaining equivalence, recall that we argued below \eqref{eq:SEV} that if $F \subseteq F'$ then~$v_F \in \overline{O_{v_{F'}}}$.
The preceding argument also implies the converse.
\end{proof}

As an immediate consequence of Proposition~\ref{prop:orbit-structure}, we get the following result, which not only reproves Lemma~\ref{lem-esupp} but also characterizes the closed orbit in an orbit closure.
For this, define
\[
  \widetilde{v} := v|_{L(v)} = v|_{\esupp(v)}.
\]

\begin{corollary}\label{cor:CO}
The orbit $O_{\tilde v}$ corresponding to the lineality space $L(v)$
is contained in every orbit closure contained in $\overline{O_v}$.
Therefore, it is the unique closed orbit contained in $\overline{O_v}$.

In particular, the orbit $O_v$ is closed if and only if $C(v) = L(v)$, i.e., $C(v)$ equals its linear span.
Moreover, $v$ is in the null cone if and only if $\esupp(v) = \emptyset$.
\end{corollary}

\section{Generating Laurent polynomials and rational invariants}\label{sec:gen-inv}
In this section, we discuss the computation of suitable rational invariants, which is the heart of our algorithms, and the main novelty of this paper.
As explained in the introduction, the starting point is the simple observation that two orbits can only be equal when they have the same support (Lemma~\ref{lem:supp-dim}).
But once we restrict to vectors of fixed support, it is natural to consider a larger class of invariants, namely Laurent polynomials, which are polynomials that can also have negative exponents.
In Section~\ref{se:lattice} we will see that the invariant Laurent polynomials for a given support naturally form a lattice that can be computed from the weight matrix.
This allows us to give an efficient algorithm for computing small sets of generators.
As a consequence, we can also efficiently compute a system of generating rational invariants.

For the rest of this section, we fix an $n$-dimensional representation $\rho_M \colon T \rightarrow \GL_n(\CC)$ of the torus~$T=(\CC^\times)^d$ given by a weight matrix $M \in \Mat_{d,n}(\ZZ)$ with columns $m^{(j)}$ for $j\in[n]$.

\subsection{Invariant Laurent polynomials}\label{se:lattice}
For $S \subseteq [n]$, consider the set of vectors with support~$S$, that is, the variety
\begin{equation}\label{eq:XS}
  X_S = \{ v \in \CC^n \;|\; \supp(v) = S \} = \{ v \in \CC^n \;|\; v_j \neq 0 \text{ if and only if } j \in S \}.
\end{equation}
The ring of regular functions on~$X_S$, denoted~$\CC[X_S]$, is naturally identified with the ring of Laurent polynomials in variables $\{x_j\}_{j\in S}$.
That is,
\[
  \CC[X_S] = \CC[x_j, x_j^{-1} \;|\; j \in S].
\]
We observe that $\rho_M$ restricts to an action of $T$ on $X_S$, and induces an action on $\CC[X_S]$.
The proposition below shows that the
algebra~$\CC[X_S]^T$ of \emph{invariant Laurent polynomials} can be succinctly described in terms of the lattice
\begin{equation}\label{def:L}
  L_S = \Bigl\{ c \in \ZZ^S \;|\; \sum_{j \in S} c_j m^{(j)} = 0 \Bigr\} = \ker(M_S) \cap \ZZ^{|S|},
\end{equation}
where $\ZZ^S := \{ c \in \RR^n \;|\; c_j = 0 \text{ for all } j \not\in S \} \cong \ZZ^{|S|}$, and $M_S$ denotes the submatrix of the weight matrix~$M$,
obtained by removing all columns except those labeled by~$S$.

\begin{proposition}\label{prop:inv-Laurent} 
\begin{enumerate}
\item Let $c\in\ZZ^S$. A Laurent monomial $x^c=\prod_{j \in S} x_j^{c_j}$ is invariant if and only if~$c \in L_S$.
\item The algebra of invariant Laurent polynomials $\CC[X_S]^T$ is spanned as a vector space by the invariant Laurent monomials.
\item If $\{c^{(1)}, c^{(2)},\dots,c^{(r)}\}$ is a lattice basis of~$L_S$, then $\CC[X_S]^T$ is generated as an algebra by the invariant Laurent monomials $\smash{\{x^{c^{(1)}},\dots,x^{c^{(r)}}\}}$.
\end{enumerate}
\end{proposition}
\begin{proof}
The first two parts are shown using an argument similar to the proof of Proposition~\ref{pro:3.1}.
The third statement is an immediate consequence.
\end{proof}

It is instructive to compare this with the discussion below Proposition~\ref{pro:3.1}, where we saw that the invariant polynomials are similarly described by the \emph{semigroup} of nonnegative vectors in the kernel of the weight matrix.
By working with vectors of fixed support, we instead obtain a natural lattice structure, which simplifies the situation considerably.
For example, the lattice~$L_S$ and hence the algebra of invariant Laurent polynomials~$\CC[X_S]^T$ have at most~$|S|\leq n$ generators -- in stark contrast to the situation for invariant polynomials.

We now discuss how to compute lattice bases as in Proposition~\ref{prop:inv-Laurent}.
It is well known that every integer matrix $M$ can be diagonalized by multiplying from left and right with unimodular matrices.
This is known as the \emph{Smith normal form}~\cite{smith}.
The Smith normal form can be computed in polynomial time~\cite{snf}.
We record these facts in the following theorem.

\begin{theorem}[Smith normal form]\label{thm:smith}
Let $M\in\Mat_{d,n}(\ZZ)$. Then, there exist unimodular matrices $U\in\Mat_{d,d}(\ZZ), W\in\Mat_{n,n}(\ZZ)$ such that
\[
  UMW=\begin{bmatrix}
\alpha_1 & 0 & 0 & & \dots & & 0\\
0 & \alpha_2 & 0 & & \dots & & 0\\
0 & 0 & \ddots & & & & 0\\
 & & & \alpha_r & & & \vdots\\
\vdots &\vdots & & & 0 & & \\
 & & & &  &\ddots & \\
0 & 0 & 0 & \dots & & & 0
\end{bmatrix}
\]
and the diagonal elements satisfy $\alpha_i\mid \alpha_{i+1}$ for $i=1,2,\dots,r-1$, where $r$ equals the rank of $M$.
The matrix $UMW$ is unique and called the Smith normal form of $M$.

Moreover, if the bit-lengths of the entries of $M$ are bounded by $b$, then the matrices $U$, $W$, and~$UMW$ can be computed in $\poly(d,n,b)$-time.
\end{theorem}

Using the Smith normal form it is easy to compute a basis of the lattice~$L_S$.
We state this in the following algorithm and corollary.

\begin{algorithm}\label{algo-lattice} \emph{Computation of a basis of the lattice of invariant Laurent monomials:}
\begin{description}
\item [Input]  $M \in \Mat_{d,n}(\ZZ)$ and $S\subseteq[n]$.
\item [Step 1] Compute the submatrix~$M_S$ of $M$ obtained by deleting all columns except those in~$S$.
\item [Step 2] Compute the Smith normal form $UM_SW$ of $M_S$ (as in Theorem~\ref{thm:smith}).
\item [Step 3] Return $\{w^{(r+1)}, w^{(r+2)}, \dots, w^{(n)}\}$, where $w^{(j)}$ denotes the $j^{th}$ column of $W$.
\end{description}
\end{algorithm}

\begin{corollary}\label{cor:smith}
Let $M\in\Mat_{d, n}(\ZZ)$ and $S \subseteq [n]$, and suppose the bit-lengths of the entries of~$M$ are bounded by~$b$.
Then Algorithm~\ref{algo-lattice} computes a basis for the lattice $L_S$ defined in \eqref{def:L} in $\poly(d,n,b)$-time.
In particular, each $w^{(j)}$ has bit-length $\poly(d,n,b)$.
\end{corollary}

Alternatively, one can use lattice algorithms; we refer the interested reader to \cite[Corollary~5.4.10]{GLS-book}.

\begin{remark}\label{rmk:easy-circuit}
It is easy to see that given an exponent vector $c = (c_1,\dots,c_n) \in \ZZ_{\ge 0}^n$, where the bit-lengths of the $c_i$s are bounded by $b$, an arithmetic circuit computing the monomial $x^c$ of size $\poly(n,b)$ can be constructed in $\poly(n,b)$-time.
Similarly, if $c\in \ZZ^n$, an arithmetic circuit with division computing the Laurent monomial $x^c$ can be constructed in $\poly(n,b)$-time.
\end{remark}

\begin{proof}[Proof of Theorem~\ref{th:laurent}]
This follows from Proposition~\ref{prop:inv-Laurent}, Corollary~\ref{cor:smith}, and Remark~\ref{rmk:easy-circuit}.
\end{proof}

\subsection{Rational invariants}\label{se:rational}
In the remainder of this section we will discuss rational invariants.
For $V = \CC^n$, recall that $\CC[V] = \CC[x_1,\dots,x_n]$ is the polynomial ring in~$n$ variables.
Let $\CC(V) = \CC(x_1,\dots,x_n)$ the field of rational functions (its fraction field).
In other words, any element in $\CC(V)$ is a ratio of two polynomials.
The action of $T$ on $\CC[V]$ extends to $\CC(V)$.
Then $\CC(V)^T$ is the field of \emph{rational invariants}.
Clearly, any invariant Laurent polynomial is a rational invariant, but the converse need not be the case.

Nevertheless, we can show that the invariant Laurent polynomials in all variables (that is, for support $S=[n]$) generate the rational invariants as a field.

\begin{proposition}\label{pro:rat}
Let $A := \CC[X_{[n]}] = \CC[x_1,x_1^{-1},\dots,x_n,x_n^{-1}]^T$ denote the algebra of invariant Laurent polynomials,
and let $F := \CC(x_1,\dots,x_n)^T$ denote the field of rational invariants.
Then, $A$~generates $F$ as a field, i.e., the field of fractions of $A$ is $F$.
\end{proposition}

\begin{proof}
Let $f \in F^\times$ and write $f = \frac p q$, where $p,q \in \CC[x_1,\dots,x_n]$ have no common factors.
Since $f$ is invariant, we have for any~$t \in T$ that 
\begin{align*}
  \frac{t \cdot p}{t \cdot q} = t \cdot f = f = \frac pq.
\end{align*}
Accordingly, $t \cdot p = \alpha(t) p$ and $t \cdot q = \alpha(t) q$ for some $\alpha(t) \in \CC^\times$.
Thus, $p$ and $q$ span one-dimensional representations.
This in turn implies that $\alpha\colon T \to \CC^\times$ is a character, as discussed in Section~\ref{subsec:rep inv tori}, and further that $p$ (and also $q$) is a sum of monomials with the same weight, i.e., $p = \sum_e p_e x^e$ such that $t \cdot x^e = \alpha(t) x^e$ for $p_e\neq0$.
In particular, $f_e = \frac q {x^e}$ is a Laurent polynomial \emph{invariant} if $p_e\ne 0$, and we can write
\begin{align*}
  f
= \frac p q
= \sum_e p_e \frac{x^e}q
= \sum_e p_e \frac1{f_e},
\end{align*}
which concludes the proof.
\end{proof}

As a direct consequence, any system of generating invariant Laurent polynomials (as an algebra) also serves as a system of generating rational invariants (as a field extension of $\CC$).
Thus we obtain:

\begin{proof}[Proof of Corollary~\ref{cor:1.3}]
This follows from Theorem~\ref{th:laurent} (with $S=[n]$) and Proposition~\ref{pro:rat}.
\end{proof}

\section{Orbit equality problem}\label{sec:oe}
In this section, we will give a polynomial time algorithm for the orbit equality problem.
Given two points, the strategy is to compute a small collection of invariant Laurent monomials (using the result of Section~\ref{sec:gen-inv})
whose evaluations at the two given points will determine whether the two points are in the same orbit.
The efficient testing of whether two Laurent monomials evaluate to the same value actually requires an idea:
this has already been studied in the literature and we briefly sketch in Section~\ref{subsec:moneq} how to do this.

We still assume an $n$-dimensional representation $\rho_M \colon T \rightarrow \GL_n(\CC)$ of the torus~$T=(\CC^\times)^d$ given by a weight matrix $M \in \Mat_{d,n}(\ZZ)$ with columns $m^{(j)}$ for $j\in[n]$.

In general, invariants can only decide orbit closure intersection, not orbit equality.
However, the crucial point is that in the varieties~\eqref{eq:XS} consisting of vectors of fixed support any $T$-orbit is closed.

\begin{proposition}\label{prop:X orbits closed}
Let $S \in [n]$, $X_S$ be the variety defined in~\eqref{eq:XS}, and $v \in X_S$.
Then the orbit $O_v$ is a closed subset of~$X_S$.
\end{proposition}
\begin{proof}
By Lemma~\ref{lem:supp-dim}~(3) we have $O_v=\overline{O_v}\cap X_S$ which implies that the orbits are closed in $X_S$.
\end{proof}

Orbit equality in~$V$ can always be reduced to orbit equality in some~$X_S$, since equality of supports is a necessary condition (Lemma~\ref{lem:supp-dim}~(1)).
The importance of the above result is that the latter orbit equality and orbit closure intersection are equivalent in $X_S$.
Together with Theorem~\ref{thm:Mumford} we obtain the following result.

\begin{corollary}\label{cor:orb sup}
Suppose $\supp(v) = \supp(w) = S$.
Then, $O_v \neq O_w$ if and only if there is an invariant Laurent monomial $f = \prod_{j \in S} x_j^{c_j}$ such that $f(v) \neq f(w)$.
\end{corollary}

Thus, we obtain the following algorithm for the orbit equality problem.

\begin{algorithm}\label{algo-oe} \emph{Deciding orbit equality:}
\begin{description}
\item [Input] $M \in \Mat_{d,n}(\ZZ)$ and $v,w \in \QQ(i)^n$.
\item [Step 1] Check if $\supp(v) = \supp(w)$. If not, $O_v \neq O_w$, so we can stop.
\item [Step 2] Use Algorithm~\ref{algo-lattice} to compute a lattice basis~$\mathcal B$ for the lattice~$L_S$ defined in~\eqref{def:L}.
\item [Step 3] For each $e \in \mathcal{B}$, we check if $v^e = w^e$ (as described in Section~\ref{subsec:moneq} below). \\
If they are all equal, then~$O_v = O_w$. Else, $O_v \neq O_w$.
\end{description}
\end{algorithm}

\begin{proof} [Proof of Theorem~\ref{thm:main}, part (1)]
The correctness of Algorithm~\ref{algo-oe} follows from Proposition~\ref{prop:inv-Laurent} and Corollary~\ref{cor:orb sup}.
We now analyze its runtime.
Clearly, the first step can be implemented efficiently.
For the second step, we can appeal to Corollary~\ref{cor:smith}.
For step~3, we first observe that, again by Corollary~\ref{cor:smith}, the exponents~$e$ have bit-length $\poly(d,n,b)$.
Then Proposition~\ref{prop:monomial} below shows that this step can also be implemented in time $\poly(d,n,b)$.
\end{proof}


\subsection{Laurent monomial equivalence}\label{subsec:moneq}
We now discuss how to test if a Laurent monomial $x^e$ evaluates to the same value at two points $v$ and $w$.
In our context, where each component $e_j$ of the exponent vector $e = (e_1,\dots,e_n)$ has poly-sized bit-lengths, it is unreasonable to evaluate the Laurent monomials explicitly, because the answer may very well require exponentially large bit-length.
Yet, it is possible to check if $v^e = w^e$ efficiently.
We describe a simple algorithm based on g.c.d.'s, which has appeared before (see, for example, \cite{ESY14}) in the case where the entries of $v$ and $w$ are in $\ZZ$ (or equivalently $\QQ$).
The result is much older; for example, it follows from the results in~\cite{BDS}, as mentioned in~\cite{Ge93}, which gives a generalization to number fields.%
\footnote{In particular Ge's result~\cite{Ge93} implies that Theorem~\ref{thm:main} extends to the case where the entries of $v$ and $w$ are taken from some algebraic number field.}
Here we present a short self-contained proof and then follow up with the rather simple extension to Gaussian rationals.

\begin{lemma}
Suppose $a_1,\dots,a_k,b_1,\dots,b_r \in\QQ$ and $e_1,\dots,e_k, f_1,\dots,f_r \in \ZZ$ have bit-lengths at most $s$.
Then, in $\poly(k,r,s)$-time, we can decide if $\prod_{i=1}^k a_i^{e_i} = \prod_{j=1}^r b_j^{f_j}$.
\end{lemma}

\begin{proof}
By clearing denominators, we may assume that $a_1,\dots,a_k,b_1,\dots,b_r$ are integers.
By moving terms to the other side, we can further assume w.l.o.g.that all $e_i, f_j \geq 0$.
Pick some $a_l$ and some~$b_m$ that are not coprime. Then, consider $d = {\rm gcd}(a_l, b_m) \geq 2$. W.l.o.g., we can assume $e_l \geq f_m$.
Then, test if $ d^{e_l - f_m} (a'_l)^{e_l}\prod_{i \neq l} a_i^{e_i}  =  (b'_m)^{f_m} \prod_{j \neq m} b_j^{f_j}$, where $a'_l = a_l/d$ and $b'_m = b_m/d$.
This is an iterative procedure which stops when each $a_i$ is coprime to $b_j$. At which point, unless all $a_i$'s and $b_j$'s are equal to $1$, both sides cannot be equal.

The question is how long does such an iterative procedure take.
Consider the quantity $P:= |a_1\cdots a_k b_1 \cdots b_r|$.
After applying one step, the resulting quantity $P'$ satisfies
$P' = P/d^2 \le P/4$. Since initially, $P$ is $2^{\poly(k,r,s)}$-sized,
there are at most a polynomial number of iterative steps.
Hence, the entire procedure takes $\poly(k,r,s)$-time.
\end{proof}

An analogous result with the same proof holds for the ring $\ZZ[i]$ of Gaussian integers and its quotient field $\QQ(i)$ of Gaussian rationals,
using that this ring has unique factorization into irreducible elements.
In the following proposition, we assume that a Gaussian rational $a = \alpha + i \beta \in \QQ(i)$ is described by giving the encodings of $\alpha$ and $\beta$ in binary.

\begin{proposition}\label{prop:monomial}
Suppose $a_1,\dots,a_k,b_1,\dots,b_r\in \QQ(i)$ and $e_1,\dots,e_k, f_1,\dots,f_r \in \ZZ$ all have bit-lengths bounded by $s$.
Then, in $\poly(k,r,s)$-time, we can decide if $\prod_{i=1}^k a_i^{e_i} = \prod_{j=1}^r b_j^{f_j}$.
\end{proposition}




\begin{remark}\label{rem:float}
For computational purposes, in many instances, numbers are described by their `floating point' representations.
The floating point description of a Gaussian rational $a \in \QQ(i)$ is described by giving the binary encodings of $\alpha,\beta \in \QQ$ and $p \in \ZZ$ such that $a = (\alpha+ i \beta) 2^p$.
If we assume that $a_1,\dots,a_k,b_1,\dots,b_r\in \QQ(i)$ in the proposition above are given by their floating point descriptions, we can still decide monomial equivalence in polynomial time.
Indeed, if we write each $a_j = (\alpha_j + i \beta_j) 2^{p_j}$ and $b_j = (\gamma_j + i \delta_j) 2^{q_j}$, then deciding whether $\prod_{j=1}^k a_j^{e_j} = \prod_{j=1}^r b_j^{f_j}$ simplifies to deciding if
\[
\left(\prod_{j=1}^k (\alpha_j + i \beta_j)^{e_j} \right) \cdot 2^{\sum_{j=1}^k e_j p_j} = \left(\prod_{j=1}^r (\gamma_j + i \delta_j)^{f_j}\right) \cdot 2^{\sum_{j=1}^r f_j q_j},
\]
which can again be interpreted as an instance of Proposition~\ref{prop:monomial} and hence can be checked in polynomial time.
Since all other computations in our algorithms only involve supports of vectors, it follows that all results in this paper generalize to this input model, as claimed in footnote~\ref{foot:float}.

An even easier special case arises for numbers of the form $a=2^p$, with $p\in\QQ$ specified by its binary encoding, as in the perfect matching application discussed in Section~\ref{subsec:applications}.
Indeed, if $a_j = 2^{p_j}$ and $b_j = 2^{q_j}$ for $j\in[n]$, then deciding whether $\prod_{j=1}^k a_j^{e_j} = \prod_{j=1}^r b_j^{f_j}$ simply amounts to verifying whether $\sum_{j=1}^n p_j e_j = \sum_{j=1}^n q_j f_j$, which is clearly possible in polynomial time.
\end{remark}

\section{Orbit closure intersection and explicit separating invariants}\label{sec:oci-to-oe}
In this section, we discuss how to solve the orbit closure intersection problem in polynomial time by efficiently reducing it to the orbit equality problem.
The problem of orbit closure intersection has a manifestly analytic point of view, but also an algebraic point of view by Mumford's theorem, Theorem~\ref{thm:Mumford}.
In other words, when orbit closures of two points do not intersect, there is an invariant polynomial that takes different values on both points, serving as a ``witness'' to the fact that the orbit closures do not intersect.
Accordingly, given two vectors whose orbit closures do not intersect, we also explain how to efficiently construct an arithmetic circuit which computes an invariant monomial separating the two vectors.

\subsection{Reduction to orbit equality}

The key idea is the following.
Recall from Theorem~\ref{th:uniqueclosedorbit} that any orbit closure~$\overline{O_v}$ contains as unique closed orbit~$O_{\tilde v}$, and that two orbit closures intersect if and only if they contain the same closed orbit.
In Corollary~\ref{cor:CO}, we showed that the unique closed orbit has a concrete polyhedral characterization: we can take $\tilde v = v|_{\esupp(v)}$, the restriction of the vector~$v$ to its essential support.
Accordingly, the map~$v\mapsto \widetilde{v}$ provides a reduction of the orbit closure intersection problem for~$\rho_M$ to the the orbit equality problem for~$\rho_M$.
The following lemma shows that the essential support (and hence the reduction map) can be computed in polynomial time by using linear programming.

\begin{lemma}\label{lem-esupp-poly}
Let $M \in\Mat_{d,n}(\ZZ)$ define an $n$-dimensional representation of $T = (\CC^\times)^d$, and let~$v\in \CC^n$.
For $k\in\supp(v)$, we have $k \in \esupp(v)$ if and only if there is a non-negative linear combination $\sum_{j \in \supp(v)} c_j m^{(j)} = 0$ such that $c_k > 0$.
If the bit-lengths of the entries of~$M$ are bounded by $b$, the latter can be decided in $\poly(d,n,b)$-time by using linear programming.
\end{lemma}

\begin{proof}
The characterization follows from Proposition~\ref{pro:3.1} and Lemma~\ref{lem-esupp}.
It amounts to a basic decisional problem of linear programming, which is well known to be solvable in polynomial time,
see~\cite{GLS-book}.
\end{proof}

The above proof also shows that a nonvanishing invariant monomial as in Lemma~\ref{lem-esupp-poly} can be computed in polynomial time.
As explained above, we arrive at the following algorithm and results.

\begin{algorithm}\label{algo:oci-to-oe} \emph{Reduction of orbit closure intersection to orbit equality:}
\begin{description}
\item[Input] $M \in \Mat_{d,n}(\ZZ), v,w \in \QQ(i)^n$.
\item[Step 1] Compute $\esupp(v)$ in the following way:
For each $k\in\supp(v)$, use linear programming to determine if there is a non-negative linear combination $\sum_{j \in \supp(v)} c_j m^{(j)} = 0$ with~$c_k>0$.
The set $\esupp(v)$ consists of all $k\in\supp(v)$ for which this is the case.
\item[Step 2] Compute $\esupp(w)$ in the same way.
\item[Step 3] Return $\tilde v = v|_{\esupp(v)}$ and $\tilde w = w|_{\esupp(w)}$.
\end{description}
\end{algorithm}

\begin{corollary} \label{cor:red-oci-oe}
Let $M \in \Mat_{d,n}(\ZZ)$ describe an $n$-dimensional representation of $T = (\CC^\times)^d$.
Further, let $v,w \in \QQ(i)^n$ and assume the bit-lengths of the entries of $M,v$, and $w$ are bounded by~$b$.
Then there is a $\poly(d,n,b)$-time reduction that reduces the problem of deciding $\overline{O_v} \cap \overline{O_w} \ne \emptyset$
to the problem of deciding if $O_{\widetilde{v}} = O_{\widetilde{w}}$, where $\widetilde{v}$ and $\widetilde{w}$ have bit-lengths bounded by $b$.
\end{corollary}

\begin{proof}[Proof of Theorem~\ref{thm:main}, part $(2)$]
This follows from part~(1), combined with Corollary~\ref{cor:red-oci-oe}.
\end{proof}

\subsection{Explicit separating invariant}
For torus actions, our reduction of orbit closure intersection to orbit equality will give us an invariant Laurent monomial that takes different values on the two points.
But a separating invariant Laurent monomial itself does \emph{not} serve as a witness (at least not naively, one needs further properties about the support of the Laurent monomial for it to serve as a witness).
We now prove Corollary~\ref{cor:sep-invar}, which asserts that given two vectors we can nevertheless efficiently construct an arithmetic circuit which computes an invariant \emph{monomial} separating them.

\begin{proof} [Proof of Corollary~\ref{cor:sep-invar}]
We already noted that, by linear programming, we can compute the essential supports of $v$ and $w$
in $\poly(d,n,b)$-time. We distinguish two cases.

\textbf{Case 1:} $\esupp(v) \ne \esupp(w)$

\noindent
Suppose $k \in \esupp(v) \setminus \esupp(w)$ without loss of generality.
By Lemma~\ref{lem-esupp} there is an invariant monomial $f = \prod_{j \in \supp(v)} x_j^{c_j}$ such that $c_k > 0$.
Let us verify that $f(v) \ne f(w)$.
We clearly have $f(v)\ne 0$.
On the other hand, $f(w)=f(\widetilde{w})=0$, since $\widetilde{w} \in \overline{O_w}$, but $k$ is not contained in $\supp(\widetilde{w})=\esupp(w)$.
So we indeed have $f(v) \ne f(w)$.
In addition, we can find $(c_1,\dots,c_n)$ in $\poly(d,n,b)$-time by linear programming (Lemma~\ref{lem-esupp-poly}),
so we can construct an arithmetic circuit for $f$ in $\poly(d,n,b)$-time by Remark~\ref{rmk:easy-circuit}.

\medskip

{\bf Case 2:} $\esupp(v) =\esupp(w)$

\noindent
Let $S := \esupp(v) =\esupp(w)$.
We assume that $\overline{O_v} \cap \overline{O_w} = \emptyset$, which implies $O_{\widetilde{v}} \cap O_{\widetilde{w}} = \emptyset$.
Thus, by Corollary~\ref{cor:orb sup}, there is an invariant Laurent monomial $f = x^e$ with the property that~$f(\widetilde v) \neq f(\widetilde w)$, and hence $f(v) \neq f(w)$.
Just like in Algorithm~\ref{algo-oe}, we can in $\poly(d,n,b)$-time compute such an exponent vector~$e\in\ZZ^n$, with bit-length of the $e_i$ bounded above by $\poly(d,n,b)$.

Our goal is to produce an invariant monomial that separates $v$ and $w$,
so we need to modify~$f$ so as to get rid of the negative exponents.
In the process, we must ensure that the bit-length of the circuit does not explode.
By Lemma~\ref{lem-esupp}, for each $k \in S$, there exists $c^{(k)}\in\NN^n$  such that $\smash{\sum_{j \in \supp(v)} c^{(k)}_j m^{(j)} = 0}$ and~$\smash{c^{(k)}_k > 0}$.
We can compute $\smash{c^{(k)}}$ in $\poly(d,n,b)$-time by linear programming.
Let $m_k = \smash{x^{c^{(k)}}}$ denote the corresponding invariant monomial.
Put $S_-:= \{j \in S\ | \ e_j < 0\}$.
If~$m_j(v) \neq m_j(w)$ for some $j \in S_-$, then $m_j$ is an explicit separating invariant monomial and we are done by Remark~\ref{rmk:easy-circuit}.
Assume now $m_j(v) = m_j(w)$ for all $j \in S_-$.
Then $\widetilde{f} := x^d := f \cdot \prod_{j\in S_-} m_j^{-e_j}$ is a Laurent monomial that separates~$v$ and~$w$.
We verify now that the exponent vector $d$ has non-negative entries.
By construction, we have for $k\in S_-$,
\begin{align*}
  d_k &= e_k + (-e_k)  c^{(k)}_k  +  \sum_{j \in S_-, j\ne k} (-e_j)\cdot c^{(j)}_k \geq 0 ,
\intertext{since $e_k<0$ and $e_j<0$ for all $j\in S_-$, while $\smash{c^{(k)}_k\geq1}$, and $\smash{c^{(j)}_k \geq 0}$.
For $k\in[n]\setminus S_-$, we have}
  d_k &= e_k + \sum_{j \in S_-} (-e_j)\cdot c^{(j)}_k \geq 0,
\end{align*}
since $e_k \geq 0$ for $k \in S \setminus S_-$ and $e_k = 0$ for $k\not\in S$, while $e_j<0$ for $j \in S_-$.
Altogether, we have shown that indeed all components of $d$ are non-negative.
We finally note that $d$ can be computed in polynomial time, in particular, it has bit-length $\poly(d,n,b)$.
So by Remark~\ref{rmk:easy-circuit}, we can construct an arithmetic circuit of size $\poly(d,n,b)$ that computes $\widetilde{f}$ in $\poly(d,n,b)$-time.
\end{proof}

\section{Orbit closure containment}\label{sec:occ}
In this section, we discuss how to solve the the orbit closure containment problem in polynomial time by efficiently reducing it to the orbit equality problem.

The notion of orbit closure containment is in general quite tricky to capture.
Polynomial invariants do not suffice, since two orbit closures can intersect (hence all polynomial invariants agree) with neither being contained in the other -- this is precisely the difference between the orbit closure intersection and the orbit closure containment problem.
Instead, the key idea for the reduction comes from one-parameter subgroups.
We already discussed in Section~\ref{sec:rep-tori} that if $w \in \overline{O_v}$ then~$O_w$ can be reached from~$v$ by a one-parameter subgroup.
The following proposition gives a concrete polyhedral description of the relevant one-parameter subgroups.

\begin{lemma} \label{lem:tori-occ-justify}
Let $M \in\Mat_{d,n}(\ZZ)$ define an $n$-dimensional representation of $T = (\CC^\times)^d$, and let~$v,w \in \CC^n$.
Then $w \in \overline{O_v}$ if and only if there exists a one-parameter subgroup $\sigma\colon \CC^\times \to T$, so $\sigma(\epsilon) = (\epsilon^{\nu_1},\dots,\epsilon^{\nu_d})$ for some $\nu\in\ZZ^d$, such that
\begin{enumerate}
\item $ \{j \in \supp(v) \ | \ m^{(j)} \cdot \nu =  0\} = \supp(w)$ and $m^{(k)} \cdot \nu > 0$ for all $k \in \supp(v) \setminus \supp(w)$;
\item $O_{(v|_{\supp(w)})} = O_w$.
\end{enumerate}
\end{lemma}

\begin{proof}
If $w \in \overline{O_v}$, then by Theorem~\ref{thm:hmtori}, we know that there is a one-parameter subgroup~$\sigma$ such that $\lim_{t \to 0} \sigma(t) v \in O_w$.
In particular this implies that $\lim_{t \to 0} \sigma(t) v$ has the same support as~$w$ and has the same orbit as $w$.
Now, both (1) and (2) follow from Lemma~\ref{lem:tori-1psg}.

For the converse, note that, again by Lemma~\ref{lem:tori-1psg}, (1) implies that $\lim_{t \to 0} \sigma(t) v = v|_{\supp(w)} \in \overline{O_v}$, hence it follows that $O_w = O_{(v|_{\supp(w)})} \subseteq \overline{O_v}$ by (2).
\end{proof}

Now, we can give our algorithm to test if $w$ is in the orbit closure of $v$.

\begin{algorithm} \label{algo:occ} \emph{Orbit closure containment:}
\begin{description}
\item[Input]  $M \in \Mat_{d,n}(\ZZ)$ and $v,w \in \QQ(i)^n$.
\item[Step 1] Check if $\supp(w) \subseteq \supp(v)$. If not, $w \notin \overline{O_v}$, so we can stop.
\item[Step 2]  Using linear programming, determine whether there exists a solution $y\in\RR^d$ to the collection of linear equalities $m^{(j)} \cdot \nu = 0$ for each $j \in \supp(w)$ and linear inequalities $m^{(k)} \cdot \nu > 0$ for all $k \in \supp(v) \setminus \supp(w)$.
If there is no solution, then $w \notin \overline{O_v}$, so we can stop.
\item[Step 3]  Use Algorithm~\ref{algo-oe} check whether $O_{(v|_{\supp(w)})} = O_w$. If yes, then $w \in \overline{O_v}$. Else, it is not.
\end{description}
\end{algorithm}

\begin{proof} [Proof of Theorem~\ref{thm:main}, part (3)]
The correctness of Algorithm~\ref{algo:occ}  follows from Lemma~\ref{lem:tori-occ-justify}.
Indeed, condition~(1) in the lemma is satisfied if and only if the algorithm passes the first two steps, and then condition~(2) is tested in the last step.

We still need to argue about the efficiency of the algorithm.
Clearly, step~1 can be done in linear time. 
Step~2 can be done in $\poly(d,n,b)$-time by linear programming.
Step~3 appeals to the orbit equality problem, which by part~(1) of the theorem can be done in $\poly(d,n,b)$-time.
\end{proof}

\section{Orbit problems for compact tori}\label{sec:compact}
So far, we have studied orbit problems for algebraic tori, that is, groups of the form $T = (\CC^\times)^d$.
In this section we consider the groups~$K = (\Ss^1)^d$, where $\Ss^1 = \{z \in \CC^\times \ | \ |z| = 1\}$.
Such groups are often called \emph{compact tori}.
Indeed, any commutative compact connected Lie group is of this form.
Besides the fundamental algorithmic interest in this setting, it is also important in applications.
For example, in physics, symmetries are often given by compact group actions, such as compact tori~\cite{guillemin1990symplectic,audin2012torus}.
We give further complexity-theoretic motivation below.

The compactness implies that orbits are closed and so the three problems in Problem~\ref{prob:main} coincide.
In this section, we show how to solve the orbit equality problem for a compact torus by reducing it to orbit equality for the corresponding algebraic torus.
Subsequently, we give an alternative reduction that works not only for tori but in fact for any connected reductive group such as~$\SL_n$.

To start, we note that it is known that any (continuous) finite-dimensional representation of $K = (\Ss^1)^d$
extends to a representation of $T = (\CC^\times)^d$ \cite{weyl:39}.
In particular, representations can be specified as before by a weight matrix $M \in\Mat_{d,n}(\ZZ)$.
Then we have the following result:

\begin{proposition}
Let $M \in\Mat_{d,n}(\ZZ)$ define an $n$-dimensional representation of $T = (\CC^\times)^d$ and $K = (\Ss^1)^d$. Let~$v,w \in \CC^n$.
Then, $O_{K,v} = O_{K,w}$ if and only if $O_{T,v} = O_{T,w}$ and $|v_j| = |w_j|$ for all $j$.
\end{proposition}

\begin{proof}
Since $K \subseteq T$, it is clear that if $O_{K,v} = O_{K,w}$, then $O_{T,v} = O_{T,w}$ and $|v_j| = |w_j|$ for all $j$.

Conversely, suppose $O_{T,v} = O_{T,w}$ and $|v_j| = |w_j|$ for all $j$.
Then, there is some $t \in T$ such that~$t \cdot v = w$.
Write $t = (t_1,\dots,t_d)$ and write each $t_i = r_i \cdot e^{\mathrm i \theta_i}$, with $r_i > 0$ and $\theta_i \in \RR$.
Then, it is easy to see that we must have $(e^{\mathrm i \theta_1},\dots,e^{\mathrm i \theta_d}) \cdot v = w$.
Thus $v$ and $w$ are in the same $K$-orbit. 
\end{proof}

\begin{proof} [Proof of Corollary~\ref{cor:main}]
We are given $M \in \Mat_{d,n}(\ZZ)$ and $v,w \in \QQ(i)^n$.
By the above proposition, we need to check if $O_{T,v} = O_{T,w}$ and if $|v_j| = |w_j|$ for all $j$.
The former can be done in polynomial time by Theorem~\ref{thm:main} and the latter can clearly be done in polynomial time.
\end{proof}

Before proceeding we give some further context and motivation.
Algorithms for the null cone membership problem (given a rational representation $\rho:G \rightarrow \GL(V)$ of a reductive group $G$ and $v \in V$, decide if $0 \in \overline{O_v}$) based on optimization methods have emerged in recent years.
They take advantage of the fact that $0 \in \overline{O_v}$ if and only if one can drive the norm to~$0$ along the orbit~$O_v$.
This can be viewed as an optimization problem where one tries to minimize (infimize) the norm along the orbit.
While this is not a convex optimization problem, it is geodesically convex by the Kempf-Ness theory~\cite{KN:78}, which allows for many of the ideas to be modified appropriately.
As far as the orbit closure intersection problem is concerned, the natural extension of this idea is as follows:
Given $v,w \in V$, first use an optimization algorithm to approximately find a point in each orbit closure with minimal norm;
let us call these points $\check{v}$, $\check{w}$.
Then, appealing to the Kempf-Ness theory again, we have that $\overline{O_v} \cap \overline{O_w} \neq \emptyset$ if and only if $\check{v}$ and $\check{w}$ are in the same orbit for a maximal compact subgroup~$K$ of $G$.
In this way, the orbit closure intersection problem for~$G$ can be reduced to the orbit equality problem for the maximal compact subgroup~$K$.
In fact, for the so-called left-right action of $\SL_n \times \SL_n$ on matrix-tuples, this idea was carried out successfully to obtain a polynomial-time algorithm for orbit closure intersection~\cite{AZGLOW}.
This further emphasizes the importance of the orbit equality problem for compact Lie group actions.

Here we report on an interesting phenomenon, which provides a kind of converse to the strategy explained above.
Namely, for any action of a connected reductive group~$G$, the orbit equality problem for the maximal compact subgroup $K\subseteq G$ is equivalent to an orbit intersection (or equality) problem for a related action of~$G$!
As this result is not crucial to the rest of the paper and requires significantly different background, we will be brief in our explanations.
We denote by~$V^*$ the contragredient or dual representation of $V$.

\begin{theorem}\label{thm:compacttogeneral}
Let $\rho\colon G \to \GL(V)$ be a finite-dimensional representation of a connected reductive group~$G$.
Let $K$ be a maximal compact subgroup of~$G$, and $\langle\cdot,\cdot\rangle$ be a $K$-invariant Hermitian inner product on~$V$.
For $v \in V$, let $\widehat{v} \in V^*$ be defined by~$\widehat{v}(w) := \langle v,w\rangle$.
Then, for $v,w \in V$, the following are equivalent:
\begin{enumerate}
\item $O_{K,v} = O_{K,w}$;
\item $O_{G,(v,\widehat{v})} = O_{G,(w,\widehat{w})}$ in $V \oplus V^*$;
\item The $G$-orbit closures of $(v,\widehat{v})$ and $(w,\widehat{w})$ in $V \oplus V^*$ intersect.
\end{enumerate}
\end{theorem}

\begin{proof}
Let $\Lie(G) \subseteq L(V)$ denote the Lie algebra of $G$.
For any linear action of~$G$ on a vector space~$U$, we get an induced action of~$\Lie(G)$ on~$U$.
Given a $K$-invariant Hermitian form $\langle\cdot,\cdot\rangle$ on~$U$, we define
the so-called moment map $\mu_U \colon U \rightarrow \Lie(G)^*$ by the formula $\mu_U(u)(X) = \left<u, X \cdot u\right>$ for~$u \in U$ and~$X \in \Lie(G)$ (up to a scalar which is not relevant for our purposes).
The celebrated Kempf-Ness theorem says that if $\mu_U(u) = 0$ then the $G$-orbit of $u$ is closed.
Moreover, it asserts that if $u'\in U$ is another point such that $\mu_U(u') = 0$, then $O_{G,u} = O_{G,u'}$ if and only if $O_{K,u} = O_{K,u'}$.

Applying the preceding to $(v,\widehat{v})$ and $(w,\widehat{w})$ in $U = V \oplus V^*$, a simple calculation shows that the moment map vanishes at either point, so the two orbits are closed.
This shows the equivalence between~(2) and~(3).
The equivalence between~(1) and~(2) follows immediately from the second part of the Kempf-Ness theorem, using that $k \widehat v = \widehat{k v}$ for any $k\in K$, since $K$ acts unitarily.
\end{proof}

\section{Concluding remarks, future directions, and open problems} \label{sec:future}
To better understand the context of our results and their potential impact on future progress, we briefly discuss some results in literature and then suggest further research directions.
In very high level, we feel that the following aspects are highlighted by this work:
the relative power and interplay between algebraic and analytical algorithms, the importance of understanding commutative actions as a stepping stone towards understanding general actions, the role of rational (as opposed to polynomial) invariants, and the subtlety of ``no go'' results, which evidently can be surpassed.

There has been an explosion of interest over the last decade in understanding invariant theory from a complexity theoretic perspective 
(we survey some of this literature in the introduction).
This rapidly developing field can be seen as an endeavor to classifying computational problems in invariant theory according to their difficulty,
finding efficient algorithms whenever possible, as well as connecting to applications in mathematics, physics, optimization, and statistics.

Invariant theory in the setting of a rational representation of a connected reductive group is the most relevant for complexity theory.
The commutative case of tori is an important special case.
Despite the well-understood structural simplicity of the corresponding invariant theory, even basic algorithmic problems are non-trivial.
Null cone membership, arguably the most basic problem, has long been known to have an efficient algorithm, as it reduces to linear programming, which non-trivially admits polynomial-time algorithms.
The problems of orbit equality, orbit closure intersection, and orbit closure containment have polynomial time algorithms, as shown in this paper.
We stress that while efficient algorithms for linear programming are ``continuous'' or ``analytic'' in nature, our algorithms use a combination of \emph{both} analytic \emph{and} algebraic techniques.
The more general problem of succinct circuits for generating polynomial invariants, which is one of the basic challenges proposed in~\cite{GCTV}, has recently shown to be impossible under natural complexity assumptions~\cite{GIMOWW}.
Yet, in this paper, we bypass this negative result, and see that \emph{rational} invariants for torus actions can be captured in a computationally efficient way without the need for succinct circuits.
It is an interesting open problem to determine if there are succinct circuits for separating invariants or null cone definers, see \cite[Problems~1.14, 1.15]{GIMOWW}.

The invariant theory of non-commutative groups has a different flavor from, and is far more complex than, the commutative case, see, for example,~\cite{Humphreys:75}.
Many interesting problems in computational invariant theory remain open in the non-commutative case.
We list a few.
First and foremost, the results in this paper motivate the investigation of the computational efficiency of systems of generating \emph{rational} invariants.
Further, it is natural to wonder if rational invariants can help capture orbit closure intersection and orbit equality for non-commutative group actions.
Another open problem is to give \emph{any} polynomial time algorithm for orbit closure intersection (and the subproblem of null cone membership).
An intermediate challenge is to ascertain whether null cone membership is in NP $\cap$ co-NP.
Note that in~\cite{BILPS:20} it is shown that the general orbit closure containment problem is NP-hard.

\subsection*{Acknowledgements}
Peter B\"urgisser and M. Levent Do\u{g}an were supported by the ERC under the European Union's Horizon 2020 research and innovation programme (grant agreement no. 787840);
Visu Makam was supported by the University of Melbourne and by NSF grant CCF-1900460.
Michael Walter acknowledges NWO Veni grant no.~680-47-459 and NWO grant OCENW.KLEIN.267.
Avi Wigderson was supported by NSF grant CCF-1900460.

\bibliographystyle{alpha}
\bibliography{refs}

\end{document}